\documentclass[showpacs,preprintnumbers,amsmath,amssymb]{revtex4}

\usepackage{graphicx}
\usepackage{bm}
\usepackage{epstopdf}
\usepackage{float}
\usepackage{dcolumn}
\usepackage{bm}
\usepackage{mathrsfs}
\usepackage{amsmath}
\newtheorem{theorem}{Theorem}[section]
\newtheorem{lemma}[theorem]{Lemma}

\newtheorem{remark}[theorem]{Remark}
\newenvironment{proof}[1][Proof.]{\begin{trivlist}
     \item[\hskip \labelsep {\bfseries #1}]}{\end{trivlist}}

\begin{document}

\title{The Dirac equation in the Kerr-de Sitter metric}
\author{D. Batic}
\email{davide.batic@uwimona.edu.jm}
\affiliation{%
Department of Mathematics,\\  University of the West Indies, Kingston 6, Jamaica 
}
\author{K. Morgan}
\email{kirk.morgan02@uwimona.edu.jm }
\affiliation{%
Department of Mathematics,\\  University of the West Indies, Kingston 6, Jamaica 
}
\author{M. Nowakowski}
\email{mnowakos@uniandes.edu.co}
\affiliation{
Departamento de Fisica,\\ Universidad de los Andes, Cra.1E
No.18A-10, Bogota, Colombia
}
\author{S. Bravo Medina}
\email{s.bravo58@uniandes.edu.co }
\affiliation{
Departamento de Fisica,\\ Universidad de los Andes, Cra.1E
No.18A-10, Bogota, Colombia
}

\date{\today}

\begin{abstract}
We consider a Fermion in the presence of a rotating black hole immersed in a universe with positive cosmological constant. After deriving new formulae for the event, Cauchy and cosmological horizons we adopt the Carter tetrad to separate the aforementioned equation into a radial and angular equation. We show how the Chandrasekhar ansatz leads to the construction of a symmetry operator that can be interpreted as the square root of the squared total angular momentum operator. Furthermore, we prove that the the spectrum of the angular operator is discrete and consists of simple eigenvalues and by means of the functional Bethe ansatz method we also derive a set of necessary and sufficient conditions for the angular operator to have polynomial solutions. Finally, we show that there exist no bound states for the Dirac equation in the non-extreme case.
\end{abstract}

\pacs{Valid PACS appear here}
\maketitle
\section{Introduction}
In this paper we study the spectral properties of the angular operator associated to massive Dirac particles outside the event horizon of a non-extreme Kerr-de Sitter (KdS) manifold and we prove the absence of bound states. The KdS metric is a solution of Einstein field equations describing an asymptotically de Sitter space-time containing a rotating black hole \cite{Mat}. Although it is not the most general model of the exterior region of a black hole we can analyze theoretically, since the black hole charge is not taken into account, it represents indeed the most realistic model in astrophysics because in general black holes are embedded in environments filled with gas and plasma and, hence any net charge is neutralized by the ambient matter. Moreover, the Wilkinson Microwave Anisotropy Probe (WMAP) indicated that our universe contains a dark energy component equivalent to a tiny positive cosmological constant \cite{WMAP1,WMAP2,WMAP3,WMAP4}. Hence, it is more than reasonable to study the Dirac equation in the geometry of a non-extreme KdS black hole. A first attempt to analyze the Dirac equation in the aforementioned metric traces back to \cite{Khan}, where the separation of the Dirac equation into an angular and a radial system was achieved by using the Kinnersley tetrad coupled with a Chandrasekhar-like ansatz for the spinor. However, no physical explanation was given concerning the symmetry operator lurking behind this so fruitful ansatz. We fill this gap by showing how the Chandrasekhar ansatz leads to the construction of such an operator that can be interpreted as the square root of the squared total angular momentum operator. Separation of the Dirac equation into ordinary differential equations in general higher dimensional Kerr-NUT-de Sitter space-times and in type D metric has been studied in \cite{Oota} and \cite{Kamran}, respectively. Furthermore, \cite{Bel} proved the essential self-adjointness of the angular operator and the absence of normalizable time-periodic solutions for the Dirac equation in the KdS metric. In that regard we extend the results of \cite{Bel} by proving the self-adjointness of the angular operator. Furthermore, we give a rigorous derivation of the spectrum of the angular operator and also derive a set of necessary and sufficient conditions so that the angular system admits polynomial solutions. Last but not least, we offer a proof of the absence of bound states which is shorter and relies on a different method than that adopted by \cite{Bel}.\\
The paper is organized as follows. In Section $1$ we give a short introduction motivating the importance of our findings. In Section $2$ we complement the results given in \cite{Mat} by giving a thorough classification of the roots of the quartic polynomial equation controlling the location of the horizons paying special attention to their algebraic multiplicities. In the non-extreme case of a KdS black hole we construct expansions for the positions of the horizons with respect to the small parameter $\Lambda$ representing the cosmological constant. We also show that there are two different types of extreme KdS black holes and for each of them we obtain analytical formulae for the horizons. We conclude this section by analyzing the naked case which is characterized by fixed values of the rotation parameter and the mass of the black hole. Also in this scenario we are able to give the exact position of the cosmological horizon. In Section $3$ we employ Carter tetrad and a Chandrasekhar-like ansatz to separate the Dirac equation into an angular and radial system. The choice of the Carter tetrad is motivated by the fact it allows to give a more elegant treatment of the separation problem and at the same time it leads to a simpler form for the radial and angular systems than those derived by \cite{Khan} who instead used the Kinnersley tetrad. We also construct a symmetry operator of the Dirac equation in the KdS metric generalizing the one obtained in \cite{Davide1} and show that it can be interpreted as the square root of the squared total angular momentum operator. In Section $4$ we study the angular eigenvalue problem. More precisely, we prove the self-adjointness of the angular operator and we show that its spectrum is purely discrete and simple by using an off-diagonalization method. Furthermore, we derive a set of necessary and sufficient conditions for the existence of polynomial solutions of the angular system. Finally, in Section $5$ we show that the Dirac equation in the KdS metric does not admit bound states solutions.

\section{The Kerr-de Sitter metric} \label{section1}
The Kerr-de Sitter (KdS) metric represents a rotating black hole immersed in an asymptotically de Sitter space-time with positive cosmological constant $\Lambda$. In Boyer-Lindquist coordinates $(t,r,\vartheta,\varphi)$ with $r>0$, $0\leq\vartheta\leq\pi$, $0\leq\varphi<2\pi$ this metric is given 
by \cite{Cart1,Cart,Khan}
\begin{equation}\label{element}
ds^{2}=\frac{\Delta_r}{\Xi^2\Sigma}(dt-a\sin^{2}\vartheta d\varphi)^{2}-\Sigma\left(\frac{dr^2}{\Delta_r}+\frac{d\vartheta^2}{\Delta_\vartheta}\right)-\frac{\Delta_\vartheta\sin^{2}\vartheta}{\Xi^2\Sigma}[adt-(r^2+a^2)d\varphi]^2
\end{equation}
with
\begin{equation*}
\Delta_r=(r^2+a^2)\left(1-\frac{\Lambda}{3}r^2\right)-2Mr,\quad 
\Delta_\vartheta=1+\frac{\Lambda a^2}{3}\cos^2{\vartheta},\quad
\Xi=1+\frac{\Lambda a^2}{3},\quad
\Sigma=r^2+a^2\cos^{2}{\vartheta}, 
\end{equation*}
where $M$ and $a$ are the mass and the angular momentum per unit mass of the black hole, respectively. By $\mathcal{F}_{KdS}(\Lambda,M,a)$ we denote the family of Kerr-de Sitter space-times. Observe that if $\Lambda=0$ the above metric reduces to the usual Kerr metric. For $\Lambda\neq 0$ and $a=0$ we obtain the line element of a Schwarzschild-de Sitter black hole. Moreover, if $\Lambda\neq 0$ and $M=0=a$ the manifold becomes a de Sitter universe with cosmological horizon at $r_c=\sqrt{3/\Lambda}$. Finally, if $M=0$ and 
$a\neq 0\neq\Lambda$ it has been shown by \cite{Mat} that the coordinate transformation 
\[
T=t/\Xi,\quad 
\overline{\varphi}=\varphi-\frac{a\Lambda}{3\Xi}t,\quad
r\cos{\vartheta}=y\cos{\Theta},\quad
y^2=\frac{r^2\Delta_\vartheta+a^2\sin^2{\vartheta}}{\Xi}
\]
maps the family of space-times $\mathcal{F}_{KdS}(\Lambda,0,a)$ to a family of de Sitter universes. Since $\Lambda>0$ it follows that $\Delta_\vartheta$ is always positive and the position of the horizons is determined by the roots of the quartic equation $\Delta_r=0$. We will use the complete root classification method developed by \cite{Arnon,Yang} in order to study the zeros of this equation. To this purpose we rewrite it as
\begin{equation}\label{**}
r^4+pr^2+qr+u=0
\end{equation}
with
\begin{equation}\label{pqu}
p=-\frac{3}{\Lambda}\left(1-\frac{\Lambda a^2}{3}\right),\quad
q=\frac{6M}{\Lambda},\quad
u=-\frac{3a^2}{\Lambda}.
\end{equation}
According to \cite{Pod} equation (\ref{**}) will have four, two, or no real roots. However, a more subtle root classification than that offered by \cite{Pod} will arise because of the different algebraic multiplicities of these roots. Furthermore, we recall that in 1998 \cite{Riess} and \cite{Perlmutter} used Type $1$a supernovae to show that the universe is accelerating, thus providing the first direct evidence that $\Lambda$ is non-zero, with $\Lambda\approx 1.7\times 10^{-121}$ Planck units. 
We will further consider $\Lambda$ as fixed. The analysis of the roots of equation (\ref{**}) is greatly simplified if we rescale the radial variable according to $\rho=r/r_{c,dS}$ where $r_{c,dS}=\sqrt{3/\Lambda}$ is the cosmological horizon of the corresponding de Sitter universe belonging to the family of space-times $\mathcal{F}_{KdS}(\Lambda,0,0)$. Then, our equation (\ref{**}) becomes
\begin{equation}\label{***}
\rho^4+\widetilde{p}\rho^2+\widetilde{q}\rho+\widetilde{u}=0
\end{equation}
with
\[
\widetilde{p}=-(1-\alpha^2),\quad
\widetilde{q}=\mu,\quad
\widetilde{u}=-\alpha^2
\]
where $\alpha=a/r_{c,dS}$ and $\mu=2M/r_{c,dS}$. Taking into account that the Schwarzschild radius $r_s=2M$ is always smaller than the de Sitter cosmological horizon we find that the parameter $\mu$ can vary only on the interval $[0,1)$. As in \cite{Arnon,Yang} we introduce the auxiliary polynomials
\begin{eqnarray*}
\delta(\widetilde{p},\widetilde{q},\widetilde{u})&=&256\widetilde{u}^3-128\widetilde{p}^2\widetilde{u}^2+144\widetilde{p}\widetilde{q}^2\widetilde{u}+16\widetilde{p}^4\widetilde{u}-27\widetilde{q}^4-4\widetilde{p}^3\widetilde{q}^2,\\
L(\widetilde{p},\widetilde{q},\widetilde{u})&=&8\widetilde{p}\widetilde{u}-9\widetilde{q}^2-2\widetilde{p}^3.
\end{eqnarray*}
Then, we have the following root classification for the quartic equation (\ref{***}).
\begin{itemize}
\item
There will be four real distinct roots if $\delta>0$, and $L>0$. These two conditions give rise to the following system of inequalities
\begin{eqnarray}
&&-\frac{27}{16}\mu^4+(1+33\alpha^2-33\alpha^4-\alpha^6)\frac{\mu^2}{4}-\alpha^2-4\alpha^4-6\alpha^6-4\alpha^8-\alpha^{10}>0\label{ineq1},\\
&&-9\mu^2+2+2\alpha^2-2\alpha^4-2\alpha^6>0.\label{ineq2}
\end{eqnarray} 
The first inequality in (\ref{ineq1}) can be cast into the form
\begin{equation}\label{LU}
0<fm^{2}+g(\alpha) m+h(\alpha)=: P_{2}(m)
\end{equation}
with $f= -9$, $g(\alpha)=(1-\alpha^2)(\alpha^4+34\alpha^2+1)$, and $h(\alpha)=-3\alpha^{2}(1+\alpha^2)^4$. Note that (\ref{LU}) is a polynomial inequality of degree two in the parameter $m:= M^2 \Lambda$. The inequality \ref{ineq2} takes the form
\begin{equation} \label{DU}
6m < 1+\alpha^{2}-\alpha^{4}-\alpha^{6}=: P_{3}(\alpha)
\end{equation}
A straightforward analysis of these inequalities is presented in Figure~\ref{giggig} which demonstrates 
the allowed region in the variables $\alpha$ and $m$ according to \eqref{LU} and \eqref{DU}. It is obvious that there exist 
extremal values for $\alpha$ and $m$ (which we will discuss below).  
\begin{figure}\label{GG}
\includegraphics[scale=0.4]{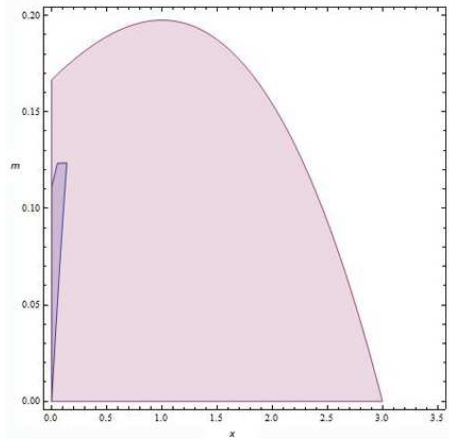}
\caption{\label{giggig}
Allowed values for the parameters $m$ and $x=2\alpha^2$. The brighter shaded region
corresponds to the inequality $L > 0$ whereas the darker shaded region to
$\delta > 0$. This demonstrates that $\delta >0$ is the more stringent inequality.
}
\end{figure}
In the same figure we see that the allowed region of the inequality $\delta>0$ lies inside the allowed region of the inequality $L>0$. 
In principle, it would suffice to consider only \eqref{LU}. It is also possible to obtain some analytical results. First of all, we observe that $P_{2}(m)$ is a concave down parabola with respect to the parameter $m$. The solution set of the inequality $P_{2}(m)>0$ will be non empty if the polynomial $P_2$ has two zeros. The condition for that reads
\begin{equation}\label{ConditionL}
g^{2}+4|f|h=(\alpha^2-4\alpha+1)^3(\alpha^2+4\alpha+1)^3> 0
\end{equation}
which will be satisfied for $\alpha<-2-\sqrt{3}$ or $-2+\sqrt{3}<\alpha<2-\sqrt{3}$ or $\alpha>2+\sqrt{3}$ and the zeroes are given by
\begin{equation}
m_{\pm}(\alpha)=\frac{-g(\alpha)\pm\sqrt{g(\alpha)^2+4|f|h(\alpha)}}{2f}=-\frac{-(1-\alpha^2)(\alpha^4+34\alpha^2+1)\pm\sqrt{(\alpha^2-4\alpha+1)^3(\alpha^2+4\alpha+1)^3}}{18}.
\end{equation}
\begin{figure}\label{FF}
\includegraphics[scale=0.4]{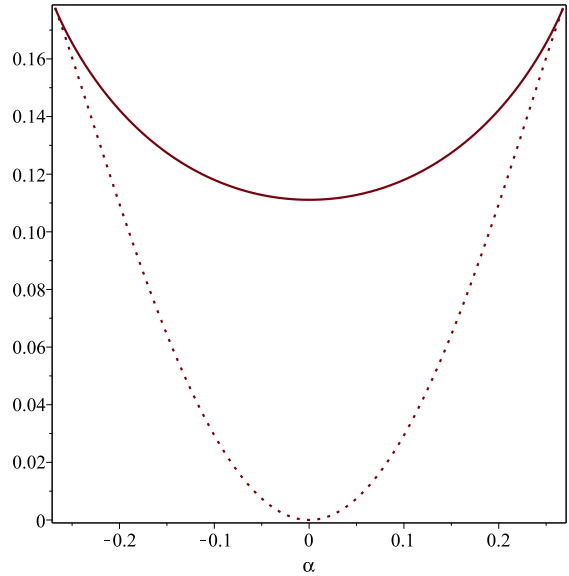}
\caption{\label{figfig}
Plot of the roots of the polynomial $P_2$ on the interval $\alpha\in(-2+\sqrt{3},2-\sqrt{3})$. The solid line corresponds to $m_-$ while the dotted line represents $m_{+}$.
}
\end{figure}
It is not difficult to verify that $m_{+}(\alpha)$ is negative if $\alpha<-2-\sqrt{3}$ or $\alpha>2+\sqrt{3}$ whereas it is positive whenever $-2+\sqrt{3}<\alpha<2-\sqrt{3}$. Since $m_{+}(\alpha)$ is even we have that on the latter interval $0<m_{+}(\alpha)<416/\sqrt{3}-240=0.177709$. 
Clearly, $m_{-}(\alpha)$ is also even and positive on the interval $-2+\sqrt{3}<\alpha<2-\sqrt{3}$ where $1/9<m_{-}(\alpha)<416/\sqrt{3}-240$.
From Fig.~\ref{figfig} we see that $m_{+}(\alpha) < m < m_{-}(\alpha)$. The maximal value of $m$ is $m_{max}=416/\sqrt{3}-240=0.177709$ and is attained at $\alpha=-2+\sqrt{3}$ and $\alpha=2-\sqrt{3}$. Furthermore, the minimal value of $m$ is $m_{min}=1/9=0.111111$ and is realized for $\alpha=0$. To summarize, from the inequalities \eqref{LU} and \eqref{ConditionL} we derive absolute limits on $a$ and $M$, more precisely
\begin{equation} \label{limitabs1}
a_{min}=-a_{max}<a<a_{max}=\frac{2\sqrt{3}-3}{\sqrt{\Lambda}}=\frac{0.464101}{\sqrt{\Lambda}}
\end{equation}
and
\begin{equation} \label{limitabs2}
M_{min}=\frac{1}{3\sqrt{\Lambda}}<M< \frac{1}{\sqrt{\Lambda}}\sqrt{\frac{416}{\sqrt{3}}-240}.
\end{equation}
Of course, given an $a$ or $\alpha$ the mass can range between $m_{max}$ and $m_{min}$ and viceversa, for a given $M$ or $m$ we get a range for $a$ or $\alpha$. Descartes' rule of signs predicts that (\ref{**}) will have one negative root at $r_{--}$ without physical meaning and three positive roots at $r_{-}$, $r_{+}$, and $r_c$ representing the Cauchy, the event, and the cosmological horizon, respectively. Note that $r_{--}<0<r_{-}<r_{+}<r_{c}$. Moreover, we have the factorization $\Delta_r=-(r-r_{--})(r-r_{-})(r-r_{+})(r-r_c)$ and $\Delta_r>0$ on the intervals $0<r<r_{-}$ and $r_{+}<r<r_c$.  Note that the presence of $r_{--}$ can always be removed since we can express this root as $r_{--}=-(r_{-}+r_{+}+r_c)$ by using one of Vieta's formulae. For completeness we give also analytical expressions for the roots of (\ref{**}) which are represented by 
\begin{eqnarray}
r_{--}&=&-\frac{\sqrt{-z_1}+\sqrt{-z_2}+\sqrt{-z_3}}{2},\quad r_{-}=\frac{\sqrt{-z_1}+\sqrt{-z_2}-\sqrt{-z_3}}{2},\label{r12}\\
r_{+}&=&\frac{\sqrt{-z_1}-\sqrt{-z_2}+\sqrt{-z_3}}{2},\quad r_{c}=\frac{-\sqrt{-z_1}+\sqrt{-z_2}+\sqrt{-z_3}}{2},\label{r34}
\end{eqnarray}
where $z_i=w_i+(2p)/3$ for every $i=1,2,3$ and 
\[
w_1=\xi_1+\xi_2,\quad w_2=\omega^2\xi_1+\omega\xi_2,\quad w_3=\omega\xi_1+\omega^2\xi_2
\]
with $\omega=\mbox{exp}(2\pi i/3)$ and 
\[
\xi_1=\sqrt[3]{-\frac{\widehat{q}}{2}+\sqrt{\left(\frac{\widehat{p}}{3}\right)^3-\left(\frac{\widehat{q}}{2}\right)^2}},\quad
\xi_2=\sqrt[3]{-\frac{\widehat{q}}{2}-\sqrt{\left(\frac{\widehat{p}}{3}\right)^3-\left(\frac{\widehat{q}}{2}\right)^2}},\quad
\widehat{p}=-\left(\frac{p^2}{3}+4u\right),\quad
\widehat{q}=\frac{2}{27}p^3-\frac{8}{3}pu+q^2.
\]
The details of the derivation of the above formulae can be found in the appendix. Note that the formulae for the roots are too unwieldy for general use. However, since $\Lambda$ can be interpreted as a small parameter, we can use perturbation theory of algebraic equations to derive simpler formulae for the Cauchy, event, and cosmological horizons. To this purpose, let us consider the polynomial
\[
P_\Lambda(r)=-\Delta_r=\frac{\Lambda}{3}r^4-\left(1-\frac{\Lambda a^2}{3}\right)r^2+2Mr-a^2.
\]
Then, we immediately see that it can be written as an unperturbed polynomial plus a perturbation containing $\Lambda$, i.e. 
$P_\Lambda(r)=T_0(r)+E_\Lambda(r)$ with $T_0(r)=-r^2+2Mr-a^2$, and $E_\Lambda(r)=\Lambda r^2(r^2+a^2)/3$. Clearly, the roots of $T_0(r)=0$ are given by $r_{\pm,K}=M\pm\sqrt{M^2-a^2}$ where $r_{+,K}$ and $r_{-,K}$ coincide with the event and Cauchy horizon of a Kerr 
black hole, respectively. Since $E_\Lambda(r)$ vanishes as $\Lambda$ tends to zero, a lemma in perturbation theory (see page $83$ in \cite{Simm}) ensures that the equation $P_\Lambda(r)=0$ has at least two roots $r_{-}(\Lambda)$ and $r_{+}(\Lambda)$ such that 
$r_{\pm}(\Lambda)\to r_{\pm,K}$ as $\Lambda\to 0$. Substituting $r_{\pm,\Lambda}=r_{\pm,K}+r_{1,\pm}\Lambda+\mathcal{O}(\Lambda^2)$ into the equation $P_\Lambda(r)=0$ yields
\[
T_0(r_{\pm,K})+\frac{1}{3}\left[r_{\pm,K}^2(r_{\pm,K}^2+a^2)-6(r_{\pm,K}-M)r_{1,\pm}\right]\Lambda+\mathcal{O}(\Lambda^2)=0.
\]
From the definition of $r_{\pm,K}$ it follows immediately that $T_0(r_{\pm,K})=0$. Moreover, the fundamental theorem of perturbation theory implies that the coefficients of the various powers of $\Lambda$ can be set to zero. Hence, we find that the coefficients $r_{1,\pm}$ are given by
\[
r_{1,\pm}=\frac{r_{\pm,K}^2(r_{\pm,K}^2+a^2)}{6(r_{\pm,K}-M)}=\pm\frac{M r_{\pm,K}^3}{3\sqrt{M^2-a^2}}.
\]
Note that for $\Lambda$ approaching zero, inequality (\ref{LU}) reduces to the condition of a non extreme Kerr black hole, that is 
$M^2>a^2$, and therefore the coefficients $r_{1,\pm}$ are always well defined. Hence, the event and Cauchy horizons of a non extreme Kerr-de Sitter black hole can be represented as follows
\[
r_{\pm}=r_{\pm,K}\pm\frac{M r_{\pm,K}^3}{3\sqrt{M^2-a^2}}\Lambda+\mathcal{O}(\Lambda^2).
\] 
Concerning the cosmological horizon we rescale the radial variable as $r=w/\sqrt{\Lambda}$ and we consider the equation
\[
P_\Lambda(\Lambda^{-1/2}w)=\frac{\Lambda^{-1}}{3}w^4+\frac{a^2}{3}w^2-\Lambda^{-1}w^2+2M\Lambda^{-1/2}w-a^2=0.
\]
Multiplying throughout by $\Lambda$ and setting $\beta=\Lambda^{1/2}$ we obtain the regular perturbation problem
\begin{equation}\label{i1}
w^4+a^2\beta^2 w^2-3w^2+6M\beta w-3a^2\beta^2=0.
\end{equation}
As $\beta\to 0$ we get the unperturbed equation $w^4-3w^2=0$ having roots at $0$ and $\pm\sqrt{3}$. Taking into account that for vanishing $a$ and $M$ the Kerr-de Sitter metric represents a de Sitter universe with cosmological horizon at $r_{c,dS}=\sqrt{3/\Lambda}$, we will look at an expansion of the form $w=\sqrt{3}+w_1\beta+w_2\beta^2+w_3\beta^3+\mathcal{O}(\beta^4)$. Substituting such an expansion into (\ref{i1}) and setting the coefficients of the various powers of $\beta$ to zero, we obtain $w_1=-M$, $w_2=-\sqrt{3}M^2/2$, and $w_3=-M(M^2-a^2)/3$. Finally, we conclude that the cosmological horizon for a non extreme Kerr-de Sitter black hole admits the following expansion in $\Lambda$, namely
\[
r_c=r_{c,dS}-M\sqrt{\Lambda}-\frac{\sqrt{3}M^2}{2} \Lambda-\frac{M}{3}(M^2-a^2)\Lambda^{3/2}+\mathcal{O}(\Lambda^2).
\]
In the case $a\to 0$ the above expression gives the position of the cosmological horizon for a Schwarzschild-de Sitter black hole.
\item
There will be two distinct and two coinciding roots whenever $L>0$, $\delta=0$ and $-\alpha_{max}<\alpha <\alpha_{max}$. 
The value of $m$ can be only on one of the two branches $m_{\pm}(\alpha)$. Since the condition $\delta=0$ reduces for $\Lambda\to 0$ to the constraint $M^2=a^2$, we will call this case the extreme Kerr-de Sitter black hole. Let $\alpha\in(-2+\sqrt{3},2-\sqrt{3})$. Since (\ref{ineq1}) is biquadratic in $\mu$, the corresponding roots are
\[
\mu_{1,\pm}(\alpha)=\pm\frac{1}{9}\sqrt{F(\alpha)+6\sqrt{G(\alpha)}},\quad
\mu_{2,\pm}(\alpha)=\pm\frac{1}{9}\sqrt{F(\alpha)-6\sqrt{G(\alpha)}}
\]
with $F(\alpha)=6g(\alpha)$, $G(\alpha)=(\alpha^2-\alpha_{-}^2)^3(\alpha^2-\alpha_+^2)^3$, and $\alpha_{\pm}=2\pm\sqrt{3}$. We have the following two cases
\begin{itemize}
\item
if $\mu=\mu_{1,+}(\alpha)$, the event and cosmological horizons coincide, i.e. $\rho_{+}=\rho_c$, and
\[
\rho_c=\frac{1}{6}\sqrt{6-6\alpha^2+6\sqrt{\alpha^4-14\alpha^2+1}},\quad
\rho_{-}=-\rho_c+\sqrt{1-\alpha^2-2\rho_c^2}.
\]
This implies that
\[
r_c=\frac{1}{6}\sqrt{\frac{3}{\Lambda}\left[6-2a^2\Lambda+2\sqrt{a^4\Lambda^2-42a^2\Lambda+9}\right]},\quad
r_{-}=-r_c+\sqrt{\frac{3}{\Lambda}-a^2-2r_c^2}
\]
with expansions in the small parameter $\Lambda$ given by
\[
r_c=\frac{1}{\sqrt{\Lambda}}-\frac{2}{3}a^2\sqrt{\Lambda}-\frac{8}{9}a^4\Lambda^{3/2}+\mathcal{O}(\Lambda^{5/2}),\quad
r_{-}=\frac{3}{2}a^2\sqrt{\Lambda}+\frac{15}{8}a^4\Lambda^{3/2}+\mathcal{O}(\Lambda^{5/2}).
\]
\item
If $\mu=\mu_{2,+}(\alpha)$, the Cauchy and event horizons coincide, i.e. $\rho_{-}=\rho_{+}$, and  
\[
\rho_{+}=\frac{1}{6}\sqrt{6-6\alpha^2-6\sqrt{\alpha^4-14\alpha^2+1}},\quad
\rho_{c}=-\rho_{+}+\sqrt{1-\alpha^2-2\rho_{+}^2},
\]
that is
\[
r_{+}=\frac{1}{6}\sqrt{\frac{3}{\Lambda}\left[6-2a^2\Lambda-2\sqrt{a^4\Lambda^2-42a^2\Lambda+9}\right]},\quad
r_{c}=-r_{+}+\sqrt{\frac{3}{\Lambda}-a^2-2r_{+}^2}
\]
with expansions in the small parameter $\Lambda$ given by
\[
r_{+}=a+\frac{2}{3}a^3\Lambda+\mathcal{O}(\Lambda^{2}),\quad
r_{c}=\sqrt{\frac{3}{\Lambda}}-a-\frac{\sqrt{3}}{2}a^2\sqrt{\Lambda}-\frac{2}{3}a^3\Lambda-\frac{41\sqrt{3}}{72}a^4\Lambda^{3/2}+\mathcal{O}(\Lambda^{2}).
\]
\end{itemize}
For the details of the derivation we refer to the appendix. From the results above we see that the extreme Kerr-de Sitter black hole in reality contains two possible scenarios. If $\mu=\mu_{1,+}$ we have $\Delta_r<0$ in the region between the Cauchy and the cosmological horizons, thus signalizing the presence of a dynamical universe where the central singularity is shielded by a Cauchy horizon. A similar scenario emerged also from the analysis of \cite{Brill} where the horizons of Reissner-Nordstr\"{o}m-de Sitter black holes have been investigated. In the limit $a\to 0$ we have $M=1/(3\sqrt{\Lambda})$ and this case corresponds to an extreme Schwarzschild-de Sitter black hole where the event and cosmological horizons coincide. In this case we have a dynamical universe exhibiting a naked central singularity. Moreover, if $\mu=\mu_{2,+}$, we have instead $\Delta_r>0$ in the region between the event and cosmological horizons. This case reduces to the usual extreme Kerr case in the limit $\Lambda\to 0$ as we can see from the above expansion for $r_{+}$. Last but not least, it is pretty amazing that we can find relatively simple formulae for the positions of the horizons in the case of an extreme Kerr-de Sitter black hole.
\item
There will be two real roots with algebraic multiplicities one and three, respectively, whenever \cite{Arnon,Yang}
\begin{equation}\label{gg}
\delta=0,\quad L=0,\quad \widetilde{p}<0,\quad \widetilde{q}\neq 0. 
\end{equation}
For $L=0$ and $\delta=0$ we have $\alpha=\alpha_{max}$ and $m=m_{max}$. The condition $\widetilde{q}\neq 0$ translates simply into the requirement $M\neq 0$ whereas the condition $\widetilde{p}<0$ is satisfied for $\Lambda a^2/3<1$. Moreover, we see that $\delta=0$ and $L=0$ at the points where the graph of $\mu_+$ intersects the graphs of $\mu_{1,+}$ and $\mu_{2,+}$. This happens for $\alpha=\pm\alpha_{-}$ with $\alpha_{-}=2-\sqrt{3}$ for which
\[
\mu_{\pm}(\pm\alpha_{-})=\mu_{1,+}(\pm\alpha_{-})=\mu_{2,+}(\pm\alpha_{-})=\frac{8}{3}\sqrt{26\sqrt{3}-45}.
\]
or equivalently, in terms of the parameters $a$ and $M$
\[
a=\pm\frac{2\sqrt{3}-3}{\sqrt{\Lambda}},\quad M=\frac{4}{3}\sqrt{\frac{78\sqrt{3}-135}{\Lambda}}.
\]
Furthermore, the roots of (\ref{***}) are $\rho_{--}=-\sqrt{6\sqrt{3}-9}$ and $\rho_{-}=\rho_{+}=\rho_c=\sqrt{(2/3)\sqrt{3}-1}$ to which correspond the following roots of (\ref{**}), namely
\[
r_{--}=-\sqrt{\frac{18\sqrt{3}-27}{\Lambda}},\quad r_{-}=r_{+}=r_{c}=\sqrt{\frac{2\sqrt{3}-3}{\Lambda}}.
\]
Taking into account that $r>0$, it is straightforward to verify that $\Delta_r>0$ for $0<r<r_c$. Hence, this scenario corresponds to a naked singularity at $r=0$ in the equatorial plane \cite{Pod} immersed in an universe with cosmological horizon. It is interesting to observe that this case does not admit a Kerr-like counterpart in the limit $\Lambda\to 0$ because in the limit of a vanishing constant the condition $L=0$ gives rise to the contradiction $1=0$.
\item
The case of two real roots both having algebraic multiplicity two which corresponds to the set of conditions $\delta=L=\widetilde{q}=0$, and $\widetilde{p}<0$ does not represent a black hole since the condition $\widetilde{q}=0$ requires that $M=0$ and therefore it will not be further analyzed. The same conclusion holds for the case $\delta<0$ corresponding to the scenario of two complex conjugate roots and two real roots each having algebraic multiplicity one.
\item
If $\delta=L=\widetilde{p}=0$ there is only one root with algebraic multiplicity four. The condition $\widetilde{p}=0$ implies that (\ref{**}) becomes $r^4+qr+u=0$ with $q=6M/\Lambda$ and $u=-9/\Lambda^2$. On the other side (\ref{**}) could be written as $(r-r_0)^4=0$. A quick comparison with $r^4+qr+u=0$ shows that $u=r_0^4=0$ which is impossible since $u=-9/\Lambda^2\neq 0$. Hence, this case can never occur.
\item
If $\delta=0$ and $L<0$ there is only one root with algebraic multiplicity two and two complex conjugate roots. This means that (\ref{**}) must be of the form $(r-r_0)^2(r^2+A^2)=0$ with $A^2>0$. Comparing the latter with (\ref{**}) yields $r_0=0$, $A^2=p$, $q=0$, and $u=0$. The last two conditions imply $M=a=0$. Hence, $p=-3/\Lambda$ and since we are interested in the case of a positive cosmological constant, this would require that $A^2<0$. Therefore, this case might be relevant in the case of an anti-de Sitter space and it will not be further investigated here.
\item
There will be no real roots when [$\delta>0$ and ($L\leq 0$ or $\widetilde{p}>0$)] or [$\delta=0$, $L= 0$ and $\widetilde{p}>0$]. The absence of real roots implies that there will be no cosmological horizon and therefore no de Sitter space asymptotically at infinity. However, this will never happen. First, $\delta \ge 0$ and $\widetilde{p} > 0$ are not compatible. From the first we obtain $\alpha^2 \le 0.0718$ (see the analysis above) whereas
the second inequality is equivalent to $\alpha^2> 1$. It remains to look into the $\delta > 0$ and $L \le 0$ case. Numerically the second inequality is satisfied for $\alpha^2 > 41/3$ which again contradicts $\delta > 0$.

\end{itemize}

Finally, after the analysis of the horizons some general comments are in order. 
Let us note that it is not always obvious that choosing one parameter in the theory constraints
the choice of other parameters. In the Kerr-de Sitter black hole this interplay happens for $a$ and $M$. 
Moreover, equations (\ref{limitabs1}) and (\ref{limitabs2}) sets absolute limits on $a$ and $M$. This seems to be a general characteristic of 
theories containing the cosmological constant $\Lambda$. Indeed, several cases can be quoted
\begin{itemize}
\item[1.]
The general condition for the existence of horizons in the Schwarzschild-de Sitter metric is \cite{me1}
\begin{equation}
\frac{1}{3\sqrt{\Lambda}} \ge M 
\end{equation}
\item[2.]
The condition for the Newtonian limit to exist in the Schwarzschild-de Sitter case is \cite{me2}
\begin{equation}
\frac{2\sqrt{2}}{3\sqrt{\Lambda}} \gg M 
\end{equation}
\item[3.]
The effective potential which determines the orbits in the Schwarzschild-de Sitter case has three
local extrema. To avoid that the first maximum coincides with the minimum one has to respect 
that the angular momentum per mass $l$ is smaller than $\sqrt{3}r_s$ with $r_s$ being the Schwarzschild radius.
On the other hand insisting that the minimum and the second maximum do not degenerate we have \cite{me3}
\begin{equation}
l < \left(\frac{3}{4}\right)^{1/3}(r_s^2r_{\Lambda})^{1/3},  \,\, r_{\Lambda}\equiv \frac{1}{\sqrt{\Lambda}}
\end{equation}
There exists a maximum radius of the order $(r_s^2r_{\Lambda})^{1/3}$ (the location of the second maximum) beyond which bound states are not possible.
\item[4.]
In the Schwarzschild-de Sitter metric different concepts of equilibrium demand a constraint on mass or density.
The hydrostatic equilibrium demands that \cite{me4}
\begin{equation}
\frac{2}{9\sqrt{\Lambda}} \ge M 
\end{equation}
The Tolman-Oppenheimer-Volkoff equation states that \cite{me4}
\begin{equation}
\Lambda < 4\pi \bar{\rho},  
\end{equation}
with $\bar{\rho}$ the average density. On the other hand the virial equilibrium applies if \cite{me5}
\begin{equation}
\rho > A\rho_{vac}, \,\, \rho_{vac}=\frac{\Lambda}{8\pi}
\end{equation}
with $\rho$ being the density and $A$ an expression which depends on
the geometry of the object.
\item[5.]
Probing the black hole evaporation via a generalized uncertainty principle one finds
that there exists a maximum temperature corresponding to a minimum black hole mass (black hole
remnant) \cite{me1} and, in case that the cosmological constant is non-zero, that we have a minimum
temperature corresponding to a maximum mass, i.e.,
\begin{eqnarray}
T_{max} &\sim& m_{pl} \leftrightarrow M_{min} \sim m_{pl} \nonumber \\
T_{min} &\sim& \sqrt{\Lambda} \leftrightarrow M_{max} \sim \frac{m_{pl}^2}{\sqrt{\Lambda}}  
\end{eqnarray}
where we have restored the Planck mass $m_{pl}$.
\end{itemize}

\section{The Dirac equation in the Kerr-de Sitter metric}
A fermion of mass $m_e$ and charge $e$ is described by the Dirac equation ($\hbar=c=G=1$) \cite{Page}
\begin{equation}\label{unopage}
\nabla^{A}_{A^{'}}\phi_{A}=\frac{m_e}{\sqrt{2}}~\chi_{A^{'}},\quad
\nabla_{A}^{A^{'}}\chi_{A^{'}}=\frac{m_e}{\sqrt{2}}~\phi_{A},  
\end{equation}
where $\nabla_{AA^{'}}$ denotes covariant differentiation, and $(\phi^A,\chi^{A^{'}})$ are the two-component spinors representing the wave function. According to \cite{New}, at each point of the space-time we can associate to the spinor basis $\xi_{a}^{A}$ a null tetrad $(\mathbf{l},\mathbf{n},\mathbf{m},\overline{\mathbf{m}})$ obeying the normalization and orthogonality relations 
\begin{equation}\label{tetrad_relations}
\mathbf{l}\cdot\mathbf{n}=1,\quad 
\mathbf{m}\cdot\overline{\mathbf{m}}=-1,\quad
\mathbf{l}\cdot\mathbf{m}=\mathbf{l}\cdot\overline{\mathbf{m}}=\mathbf{n}\cdot\mathbf{m}=\mathbf{n}\cdot\overline{\mathbf{m}}=0.
\end{equation}
Moreover, to any tetrad we can associate a unitary spin-frame $(o^A,\iota^A)$ defined uniquely up to an overall sign factor by the relations $o^A\overline{o}^{A^{'}}=l^a$, $\iota^A\overline{\iota}^{A^{'}}=n^a$, $o^A\overline{\iota}^{A^{'}}=m^a$, $\iota^A\overline{o}^{A^{'}}=\overline{m}^a$, and $o^A \iota^A=1$ \cite{pen}. As in \cite{Davide1} we denote by $\phi_0$ and $\phi_1$ the components of $\phi^A$ in the spin-frame $(o^A,\iota^A)$, and by $\chi_{0^{'}}$ and $\chi_{1^{'}}$ the components of $\chi^{A^{'}}$ in $(\overline{o}^A,\overline{\iota}^A)$, more precisely $\phi_0=\phi_A o^A$, $\phi_1=\phi_A\iota^A$, $\chi_{0^{'}}=\chi_{A^{'}}\overline{o}^{A^{'}}$, and $\chi_{1^{'}}=\chi_{A^{'}}\overline{\iota}^{A^{'}}$, and we introduce the spinor $\psi=(F_1,F_2,G_1,G_2)^T$ with components $F_1=-\phi_0$, $F_2=\phi_1$, $G_1=i\chi_{1^{'}}$, and $G_2=i\chi_{0^{'}}$. Then, the two equations in (\ref{unopage}) can be written in matrix form as
\begin{equation} \label{uno}
\mathcal{O}_{D}\Psi=0,\quad
\mathcal{O}_{D}=\left(\begin{array}{cccc} 
-m_{e} &  0     & \alpha_{+} & \beta_{+}  \\
0      & -m_{e} & \beta_{-}  & \alpha_{-} \\
\widetilde{\alpha}_{-} & -\widetilde{\beta}_{+} & -m_{e} & 0\\
-\widetilde{\beta}_{-} & \widetilde{\alpha}_{+} &   0 & -m_{e}
\end{array}\right),
\end{equation}
with
\begin{eqnarray*}
\alpha_{+}&=&-i\sqrt{2}(D+\overline{\epsilon}-\overline{\rho}),\quad
\beta_{+}=i\sqrt{2}(\delta+\overline{\pi}-\overline{\alpha}),\quad
\beta_{-}=i\sqrt{2}(\overline{\delta}+\overline{\beta}-\overline{\tau}),\quad
\alpha_{-}=-i\sqrt{2}(\widetilde{\Delta}+\overline{\mu}-\overline{\gamma}),\\
\widetilde{\alpha}_{-}&=&-i\sqrt{2}(\widetilde{\Delta}+\mu-\gamma),\quad
\widetilde{\beta}_{+}=i\sqrt{2}(\delta+\beta-\tau),\quad
\widetilde{\beta}_{-}=i\sqrt{2}(\overline{\delta}+\pi-\alpha),\quad
\widetilde{\alpha}_{+}=-i\sqrt{2}(D+\epsilon-\rho),
\end{eqnarray*}
where $\mu$, $\gamma$, $\beta$, $\tau$, $\rho$, $\pi$, $\alpha$ are the spin coefficients and $D$, $\widetilde{\Delta}$, $\delta$ are the directional derivatives along the tetrad $(\mathbf{l},\mathbf{n},\mathbf{m},\overline{\mathbf{m}})$. In what follows we consider the Dirac equation in the presence of a non-extreme Kerr-de Sitter black hole. The Dirac equation in this geometry was computed and separated in \cite{Khan} with the help of the Kinnersley tetrad \cite{Kinn}. In view of the separation of the Dirac equation we choose to work with a Carter tetrad \cite{Carter} which allows for a more elegant treatment of the separation problem and leads to simpler forms of the radial and angular equations than those derived in \cite{Khan}. This symmetric null tetrad also generalizes the one employed in \cite{Davide1} where the Dirac equation in the Kerr metric has been investigated. The line element of the Kerr-de Sitter metric given by (\ref{element}) suggests that we introduce the differential forms
\[
\vartheta_\mu^1=\sqrt{\frac{\Sigma}{\Delta_r}}~dr,\quad
\vartheta_\mu^2=\sqrt{\frac{\Sigma}{\Delta_\vartheta}}~d\vartheta,\quad
\vartheta_\mu^3=\frac{\sin{\vartheta}}{\Xi}\sqrt{\frac{\Delta_\vartheta}{\Sigma}}[(r^2+a^2)d\varphi-a dt],\quad
\vartheta_\mu^4=\frac{1}{\Xi}\sqrt{\frac{\Delta_r}{\Sigma}}[dt-a\sin^2{\vartheta}~d\varphi].
\]
With the help of $(5.119)$ in \cite{Carter} we can construct a symmetric null tetrad as follows
\[
l_\mu=\frac{\vartheta_\mu^4+\vartheta_\mu^1}{\sqrt{2}},\quad
n_\mu=\frac{\vartheta_\mu^4-\vartheta_\mu^1}{\sqrt{2}},\quad
m_\mu=\frac{\vartheta_\mu^2+i\vartheta_\mu^3}{\sqrt{2}}.
\]
Hence, we have
\begin{eqnarray*}
l_\mu&=&\left(\frac{1}{\Xi}\sqrt{\frac{\Delta_r}{2\Sigma}},\sqrt{\frac{\Sigma}{2\Delta_r}},0,-\frac{a\sin^2{\vartheta}}{\Xi}\sqrt{\frac{\Delta_r}{2\Sigma}}\right),\quad
l^\mu=\left(\frac{\Xi(r^2+a^2)}{\sqrt{2\Sigma\Delta_r}},-\sqrt{\frac{\Delta_r}{2\Sigma}},0,\frac{\Xi a}{\sqrt{2\Sigma\Delta_r}}\right),\\
n_\mu&=&\left(\frac{1}{\Xi}\sqrt{\frac{\Delta_r}{2\Sigma}},-\sqrt{\frac{\Sigma}{2\Delta_r}},0,-\frac{a\sin^2{\vartheta}}{\Xi}\sqrt{\frac{\Delta_r}{2\Sigma}}\right),\quad
n^\mu=\left(\frac{\Xi(r^2+a^2)}{\sqrt{2\Sigma\Delta_r}},\sqrt{\frac{\Delta_r}{2\Sigma}},0,\frac{\Xi a}{\sqrt{2\Sigma\Delta_r}}\right),\\
m_\mu&=&\left(\frac{ia\sin{\vartheta}}{\Xi}\sqrt{\frac{\Delta_\vartheta}{2\Sigma}},0,\sqrt{\frac{\Sigma}{2\Delta_\vartheta}},-i\sin{\vartheta}\frac{r^2+a^2}{\Xi}\sqrt{\frac{\Delta_\vartheta}{2\Sigma}}\right),\quad
m^\mu=\left(-\frac{i\Xi a\sin{\vartheta}}{\sqrt{2\Sigma\Delta_\vartheta}},0,-\sqrt{\frac{\Delta_\vartheta}{2\Sigma}},-\frac{i\Xi}{\sin{\vartheta}\sqrt{2\Sigma\Delta_\vartheta}}\right).
\end{eqnarray*}
It is not difficult to verify that this tetrad satisfies the conditions in (\ref{tetrad_relations}) and it is made of null vectors, i.e. $l_\mu l^\mu=n_\mu n^\mu=m_\mu m^\mu=\overline{m}_\mu\overline{m}^\mu=0$. Using the above tetrad and $(2.3\mbox{a})-(2.3\mbox{d})$ in \cite{Davide1} the spin coefficients are found to be $\kappa=\sigma=\lambda=\nu=0$, $\rho=\mu$, $\gamma=\epsilon$, $\beta=-\alpha$, $\pi=-\tau$ and
\[
\mu=\frac{1}{\widetilde{\rho}}\sqrt{\frac{\Delta_r}{2\Sigma}},\quad
\epsilon=\frac{1}{2}\left(\mu-\frac{\Delta^{'}_r}{2\sqrt{2\Sigma\Delta_r}}\right),\quad
\alpha=\frac{1}{2}\left[\frac{\dot{\Delta}_\vartheta}{2\sqrt{2\Sigma\Delta_\vartheta}}+\sqrt{\frac{\Delta_\vartheta}{2\Sigma}}\left(\cot{\vartheta}+\frac{ia\sin{\vartheta}}{\widetilde{\rho}}\right)\right],\quad
\tau=-\frac{ia\sin{\vartheta}}{\widetilde{\rho}}\sqrt{\frac{\Delta_\vartheta}{2\Sigma}},
\]
where $\widetilde{\rho}=r+ia\cos{\vartheta}$, prime denotes differentiation with respect to the radial variable and dot means differentiation with respect to the angular variable $\vartheta$. We are now ready to give a more explicit form to the Dirac equation (\ref{uno}). If we introduce the following operators 
\[
\mathcal{D}_{\pm}=\frac{\partial}{\partial r}\mp\frac{\Xi}{\Delta_r}\left[(r^2+a^2)\frac{\partial}{\partial t}+a\frac{\partial}{\partial\varphi}\right],\quad  
\mathcal{L}_{\pm}=\frac{\partial}{\partial \vartheta}+\frac{1}{2}\cot{\vartheta}\mp\frac{i\Xi}{\Delta_\vartheta}\left(a\sin{\vartheta}\frac{\partial}{\partial t}+\csc{\vartheta}\frac{\partial}{\partial\varphi}\right),
\]
the entries of the matrix in (\ref{uno}) are computed to be
\begin{eqnarray*}
\alpha_{\pm}&=&\pm i\sqrt{\frac{\Delta_r}{\Sigma}}(\mathcal{D}_{\pm}+f(r,\vartheta)),\quad 
\beta_{\pm}=-i\sqrt{\frac{\Delta_\vartheta}{\Sigma}}(\mathcal{L}_{\mp}+g(r,\vartheta)),\\
\widetilde{\alpha}_{\pm}&=&\pm i\sqrt{\frac{\Delta_r}{\Sigma}}(\mathcal{D}_{\pm}+\overline{f(r,\vartheta)}),\quad 
\widetilde{\beta}_{\pm}=-i\sqrt{\frac{\Delta_\vartheta}{\Sigma}}(\mathcal{L}_{\mp}+\overline{g(r,\vartheta)}),\\
\end{eqnarray*}
with
\[
f(r,\vartheta)=\frac{1}{2}\left(\frac{1}{\overline{\widetilde{\rho}}}+\frac{\Delta^{'}_r}{2\Delta_r}\right),\quad
g(r,\vartheta)=\frac{1}{2}\left(\frac{\dot{\Delta}_\vartheta}{2\Delta_\vartheta}+\frac{ia\sin{\vartheta}}{\overline{\widetilde{\rho}}}\right).
\]
As in \cite{Davide1}, we replace \eqref{uno} by a modified but equivalent equation
\begin{equation} \label{vdoppio}
\mathcal{W}\widehat{\psi}(t,r,\vartheta,\varphi)=0,\quad \mathcal{W}=\Gamma S^{-1}\mathcal{O}_{D}S,
\end{equation}
where $\widehat{\psi}=S^{-1}\Psi=(\hat{F}_{1},\hat{F}_{2},\hat{G}_{1},\hat{G}_{2})^{T}$ and $\Gamma$ and $S$ are non singular $4\times 4$ matrices, whose elements may depend on the variables $r$ and $\vartheta$. Proceeding as in Lemma 2.1 in \cite{Davide1} it can be showed that for $r_{+}<r<r_c$ there exist non singular $4\times 4$ matrices
\begin{equation}\label{S}
S=(\Delta_r\Delta_\vartheta)^{-1/4}\mbox{diag}(\widetilde{\rho}^{-1/2},\widetilde{\rho}^{-1/2},\overline{\widetilde{\rho}}^{-1/2},\overline{\widetilde{\rho}}^{-1/2}),\quad
\Gamma=\mbox{diag}(-i\overline{\widetilde{\rho}},i\overline{\widetilde{\rho}},i\widetilde{\rho},-i\widetilde{\rho})
\end{equation}
with $\mbox{det}(S)=(\Sigma\Delta_r\Delta_\vartheta)^{-1}$ and $\mbox{det}(\Gamma)=\Sigma^2$ such that the operator $\mathcal{W}$ decomposes into the sum of an operator containing only derivatives respect to the variables $t$, $r$ and $\varphi$ and of an operator involving only derivatives respect to $t$, $\vartheta$ and $\varphi$. More precisely, we have
\begin{equation} \label{unoo}
\mathcal{W}=\mathcal{W}_{(t,r,\varphi)}+\mathcal{W}_{(t,\vartheta,\varphi)}
\end{equation}
with
\begin{eqnarray*}
\mathcal{W}_{(t,r,\varphi)}&=&\left( \begin{array}{cccc}
                            im_{e}r&0&\sqrt{\Delta_r}\mathcal{D}_{+}&0\\
                            0&-im_{e}r&0&\sqrt{\Delta_r}\mathcal{D}_{-}\\
                            \sqrt{\Delta_r}\mathcal{D}_{-}&0&-im_{e}r&0\\
                             0&\sqrt{\Delta_r}\mathcal{D}_{+}&0&im_{e}r
                            \end{array} \right),\\
\mathcal{W}_{(t,\vartheta,\varphi)}&=&\left( \begin{array}{cccc}
                            am_{e}\cos{\vartheta}&0&0&-\sqrt{\Delta_\vartheta}\mathcal{L}_{-}\\
                            0&-am_{e}\cos{\vartheta}&\sqrt{\Delta_\vartheta}\mathcal{L}_{+}&0\\
                            0&-\sqrt{\Delta_\vartheta}\mathcal{L}_{-}&am_{e}\cos{\vartheta}&0\\
                            \sqrt{\Delta_\vartheta}\mathcal{L}_{+}&0&0&-am_{e}\cos{\vartheta}
                            \end{array} \right).
\end{eqnarray*}
\begin{remark}
It should be noted that this decomposition continues to hold even for extreme Kerr-de Sitter black holes. In addition, since the Kerr - de Sitter metric goes over into the de Sitter metric for $M=a=0$, we obtain from the decomposition above a similar result for the Dirac equation in a de Sitter universe. Taking into account that in this case $\widehat{\Delta}_\vartheta=1$ and $\widehat{\Delta}_r=r^2(1-\Lambda r^2/3)$ with $r\in(0,\sqrt{3/\Lambda})$ we can find non singular $4\times 4$ matrices
\[
\widehat{S}=\Delta_r^{-1/4} \mbox{diag}(r^{-1/2},r^{-1/2},r^{-1/2},r^{-1/2}),\quad
\widehat{\Gamma}= \mbox{diag}(-ir,ir,ir,-ir)
\]
with $\mbox{det}(\widehat{S})=(r^2\widehat{\Delta}_r)^{-1}$ and $\mbox{det}(\widehat{\Gamma})=r^4$ such that the operator $\mathcal{W}^{(\mbox{dS})}$ decomposes as follows $\mathcal{W}^{(\mbox{dS})}=\mathcal{W}_{(t,r,\varphi)}^{(\mbox{dS})}+\mathcal{W}_{(t,\vartheta,\varphi)}^{(\mbox{dS})}$ with
\[
\mathcal{W}_{(t,r,\varphi)}^{(\mbox{dS})}=\left( \begin{array}{cccc}
                            im_{e}r&0&\sqrt{\widehat{\Delta}_r}\widehat{\mathcal{D}}_{+}&0\\
                            0&-im_{e}r&0&\sqrt{\widehat{\Delta}_r}\widehat{\mathcal{D}}_{-}\\
                            \sqrt{\widehat{\Delta}_r}\widehat{\mathcal{D}}_{-}&0&-im_{e}r&0\\
                             0&\sqrt{\widehat{\Delta}_r}\widehat{\mathcal{D}}_{+}&0&im_{e}r
                            \end{array} \right),\quad
\mathcal{W}_{(t,\vartheta,\varphi)}^{(\mbox{dS})}=\left( \begin{array}{cccc}
                            0&0&0&-\widehat{\mathcal{L}}_{-}\\
                            0&0&\widehat{\mathcal{L}}_{+}&0\\
                            0&-\widehat{\mathcal{L}}_{-}&0&0\\
                            \widehat{\mathcal{L}}_{+}&0&0&0
                            \end{array} \right),
\]
where
\[
\widehat{\mathcal{D}}_{\pm}=\frac{\partial}{\partial r}\mp\frac{1}{1-\frac{\Lambda r^2}{3}}\frac{\partial}{\partial t},\quad  
\widehat{\mathcal{L}}_{\pm}=\frac{\partial}{\partial \vartheta}+\frac{1}{2}\cot{\vartheta}\mp\csc{\vartheta}\frac{\partial}{\partial\varphi},
\]
\end{remark}
We compute now the commutation relations needed to construct a symmetry operator of the Dirac equation in the Kerr-de Sitter metric generalizing the one obtained in \cite{Davide1} for the same equation in the Kerr metric. To this purpose, let $\widehat{\psi}\in C^{(2)}(\Omega)$ with $\Omega=\mathbb{R}\times(r_+,r_c)\times[0,\pi)\times[0,2\pi]$. Then, it is not difficult to verify that the following commutators hold
\begin{equation*}
\left[\mathcal{W}_{(t,r,\varphi)},\mathcal{W}_{(t,\vartheta,\varphi)}\right]=0,\quad \left[\mathcal{W}_{(t,r,\varphi)},\mathcal{W}\right]=\left[\mathcal{W}_{(t,\vartheta,\varphi)},\mathcal{W}\right]=0.
\end{equation*} 
Moreover, the matrix $\Gamma$ splits into the sum $\Gamma=\Gamma_{(r)}+\Gamma_{(\vartheta)}$ with $\Gamma_{(r)}=i\mbox{diag}(-r,r,r,-r)$ and $\Gamma_{(\vartheta)}=a~\mbox{diag}(-\cos{\vartheta},\cos{\vartheta},-\cos{\vartheta},\cos{\vartheta})$ satisfying the commutation relations
\begin{equation*}
\left[\Gamma_{(r)},\Gamma_{(\vartheta)}\right]=0,\quad\left[\Gamma_{(r)},\mathcal{W}_{(t,\vartheta,\varphi)}\right]=\left[\Gamma_{(\vartheta)},\mathcal{W}_{(t,r,\varphi)}\right]=0.
\end{equation*}
Since the Kerr-de Sitter metric is axially symmetric, it is natural to make the following ansatz for the spinors $\widehat{\psi}$ entering in \eqref{vdoppio}, namely
\begin{equation} \label{psi}
\widehat{\psi}(t,r,\vartheta,\varphi)=e^{i\omega t}e^{i\left(k+\frac{1}{2}\right)\varphi}\widetilde{\psi}(r,\vartheta)
\end{equation}
where $\omega$ and $k\in\mathbb{Z}$ are the energy and the azimuthal quantum number of the particle,respectively, and $\widetilde{\psi}(r,\vartheta)\in\mathbb{C}^{4}$. Inserting \eqref{psi} in \eqref{unoo}, it can be verified that $\widetilde{\psi}(r,\vartheta)$ satisfies the equation
\begin{equation} \label{mod}
\left(\mathcal{W}_{(r)}+\mathcal{W}_{(\vartheta)}\right)\widetilde{\psi}=0,
\end{equation}
where
\begin{eqnarray}
\mathcal{W}_{(r)}&=&\left( \begin{array}{cccc}
                            im_{e}r&0&\sqrt{\Delta_r}\widetilde{\mathcal{D}}_{+}&0\\
                            0&-im_{e}r&0&\sqrt{\Delta_r}\widetilde{\mathcal{D}}_{-}\\
                            \sqrt{\Delta_r}\widetilde{\mathcal{D}}_{-}&0&-im_{e}r&0\\
                             0&\sqrt{\Delta_r}\widetilde{\mathcal{D}}_{+}&0&im_{e}r
                            \end{array} \right) \label{mod1}\\
\mathcal{W}_{(\vartheta)}&=&\left( \begin{array}{cccc}
                            am_{e}\cos{\vartheta}&0&0&-\sqrt{\Delta_\vartheta}\widetilde{\mathcal{L}}_{-}\\
                            0&-am_{e}\cos{\vartheta}&\sqrt{\Delta_\vartheta}\widetilde{\mathcal{L}}_{+}&0\\
                            0&-\sqrt{\Delta_\vartheta}\widetilde{\mathcal{L}}_{-}&am_{e}\cos{\vartheta}&0\\
                            \sqrt{\Delta_\vartheta}\widetilde{\mathcal{L}}_{+}&0&0&-am_{e}\cos{\vartheta}
                            \end{array} \right) \label{mod2}
\end{eqnarray}
with
\[
\widetilde{\mathcal{D}}_{\pm}=\frac{\partial}{\partial r}\mp i\frac{\Xi}{\Delta_r}\left[\omega(r^2+a^2)+a\widehat{k}\right],\quad
\widetilde{\mathcal{L}}_{\pm}=\frac{\partial}{\partial \vartheta}+A_\pm(\vartheta)
\]
where
\[
A_\pm(\vartheta)=\frac{1}{2}\cot{\vartheta}\pm\frac{\Xi}{\Delta_\vartheta}\left(a\omega\sin{\vartheta}+\frac{\widehat{k}}{\sin{\vartheta}}\right),\quad\widehat{k}=k+\frac{1}{2}.
\]
At this point it is instructive to compare the above expressions with the operators $\mathcal{D}_0$, $\mathcal{D}_{1/2}$, $\mathcal{L}_{1/2}$, and $\mathcal{L}^\dagger_{1/2}$ represented by equations $(2.6)-(2.9)$ in \cite{Davide3} where the separation of the Dirac equation in the Kerr-de Sitter metric was achieved by using the Kinnersley tetrad. It can be immediately observed that the major benefit of using the Carter tetrad is to transform away the terms 
\[
\frac{1}{2\Delta_r}\frac{d\Delta_r}{dr},\quad \frac{1}{2\sqrt{\Delta_\vartheta}\sin{\vartheta}}\frac{d}{d\vartheta}(\sqrt{\Delta_\vartheta}\sin{\vartheta})
\]
appearing in $(2.6)-(2.9)$ in \cite{Davide3}. This in turn will lead to a very simplified form of the radial and angular systems. Let us set
\begin{equation} \label{Chandra1}
\widetilde{\psi}(r,\vartheta)=(R_{-}(r)S_{+}(\vartheta),R_{+}(r)S_{-}(\vartheta),R_{+}(r)S_{+}(\vartheta),R_{-}(r)S_{-}(\vartheta))^{T}.
\end{equation}
According to Chandrasekhar Ansatz \cite{Chandra} and proceeding as in \cite{Davide1} equation \eqref{mod} splits into the following two systems of linear first order differential equations, namely 
\begin{equation} \label{radial} 
\left( \begin{array}{cc}
     \sqrt{\Delta_r}\widetilde{\mathcal{D}}_{-}&-im_{e}r-\lambda\\
     im_{e}r-\lambda&\sqrt{\Delta_r}\widetilde{\mathcal{D}}_{+}
           \end{array} \right)\left( \begin{array}{cc}
                                     R_{-} \\
                                     R_{+}
                                     \end{array}\right)=0,
\end{equation}
\begin{equation} \label{angular}
\left( \begin{array}{cc}
     -\sqrt{\Delta_\vartheta}\widetilde{\mathcal{L}}_{-} & \lambda+am_{e}\cos{\vartheta}\\
                \lambda-am_{e}\cos{\vartheta} & \sqrt{\Delta_\vartheta}\widetilde{\mathcal{L}}_{+}
           \end{array} \right)\left( \begin{array}{cc}
                                     S_{-} \\
                                     S_{+}
                                     \end{array}\right)=0
\end{equation} 
where $\lambda$ is a separation constant. Note that when $a=\Lambda=0$ the angular eigenfunctions $S_{\pm}$ reduce to the well-known spin-weighted spherical harmonics whereas for $a\neq 0=\Lambda$ the same eigenfunctions satisfy a Heun equation \cite{Davide2}. Finally, in the general case of non vanishing values of the parameters $\Lambda$ and $a$ it has been shown in \cite{Davide3} that $S_{\pm}$ obey a generalized Heun equation (GHE). Since the GHE has been scarcely studied in the mathematical literature made exceptions of \cite{S1} and \cite{S2}, the next section will be devoted to study the spectrum of the angular eigenvalue problem, the dependence of the eigenvalues upon the relevant physical parameters, and to obtain series representations for the eigenfunctions. We conclude this section by giving a physical interpretation to the separation constant. To this purpose we will use the Chandrasekhar ansatz to generate a new operator $J$ and we will show that it commutes with the Dirac operator $\mathcal{O}_D$ in the Kerr-de Sitter metric, thus being a symmetry operator for $\mathcal{O}_D$. It will turn out that $\lambda$ can be seen as an eigenvalue of the operator $J$ whose interpretation will emerge from taking the limit $\Lambda\to 0$ in the expression for $J$. Since this limit coincides with the square root of the squared total angular momentum for a Dirac particle in the Kerr metric obtained in \cite{Davide1}, we can conclude that $J$ is the squared total angular momentum for a Dirac particle in the Kerr-de Sitter metric. Proceeding as in \cite{Davide1} we can construct a matrix operator
\[
\widehat{J}=\left(
\begin{array}{cccc}
0&0&-ia\widetilde{\rho}\cos{\vartheta}\frac{\sqrt{\Delta_r}}{\Sigma}\mathcal{D}_+&r\widetilde{\rho}\frac{\sqrt{\Delta_\vartheta}}{\Sigma}\mathcal{L}_{-}\\
0&0&-r\widetilde{\rho}\frac{\sqrt{\Delta_\vartheta}}{\Sigma}\mathcal{L}_{+}&-ia\widetilde{\rho}\cos{\vartheta}\frac{\sqrt{\Delta_r}}{\Sigma}\mathcal{D}_{-}\\
ia\overline{\widetilde{\rho}}\cos{\vartheta}\frac{\sqrt{\Delta_r}}{\Sigma}\mathcal{D}_{-}&r\overline{\widetilde{\rho}}\frac{\sqrt{\Delta_\vartheta}}{\Sigma}\mathcal{L}_{-}&0&0\\
-r\overline{\widetilde{\rho}}\frac{\sqrt{\Delta_\vartheta}}{\Sigma}\mathcal{L}_{+}&ia\overline{\widetilde{\rho}}\cos{\vartheta}\frac{\sqrt{\Delta_r}}{\Sigma}\mathcal{D}_{+}&0&0
\end{array}
\right)
\]
such that $\widehat{J}\widehat{\psi}=\lambda\widehat{\psi}$. Let $J=S\widehat{J}S^{-1}$ with $S$ defined as in (\ref{S}). Then, 
\[
J=S\Gamma^{-1}(\Gamma_{(\vartheta)}\mathcal{W}_{(r)}-\Gamma_{(r)}\mathcal{W}_{(\vartheta)})S^{-1}
\]
is a symmetry operator for the formal Dirac operator $\mathcal{O}_D$ since $[\mathcal{O}_D,J]=0$. We skip the proof because the method and the computation is essentially the same as those appearing in Lemma~$2.4$ in \cite{Davide1}. Finally, letting $\Lambda\to 0$ in the above expression reproduces the square root of the squared total angular momentum for a Dirac particle in the Kerr metric obtained in \cite{Davide1}.

\section{The angular eigenvalue problem}
We start by observing that the angular eigenvalue problem (\ref{angular}) admits the discrete symmetry $P:\vartheta\longrightarrow\pi-\vartheta$ so that
\[
PS_{+}(\vartheta)=S_{+}(\pi-\vartheta)=S_{-}(\vartheta),\quad PS_{-}(\vartheta)=S_{-}(\pi-\vartheta)=S_{+}(\vartheta).
\] 
This means that if we decide to eliminate $S_{+}$ in favor of $S_{-}$ in (\ref{angular}) to get a second order differential equation for $S_{-}$ the corresponding equation for $S_{+}$ can be obtained by applying the transformation $P$ to the equation satisfied by $S_{-}$. It is not difficult to verify that the equation satisfied by $S_{-}$ is
\begin{equation}\label{Smeno}
\sqrt{\Delta_\vartheta}\widetilde{\mathcal{L}}_{+}(\sqrt{\Delta_\vartheta}\widetilde{\mathcal{L}}_{-}S_{-})+
\frac{am_e\Delta_\vartheta\sin{\vartheta}}{\lambda+am_e\cos{\vartheta}}\widetilde{\mathcal{L}}_{-}S_{-}+
(\lambda^2-a^2 m_e^2\cos^2{\vartheta})S_{-}=0.
\end{equation}
\begin{remark}
Note that in the last term of equation $(3.39)$ in \cite{Khan} and equation $(3.2)$ in \cite{Davide3} there is a typo and the term $\lambda^2+a^2 m_e^2\cos^2{\vartheta}$ should be replaced by $\lambda^2-a^2 m_e^2\cos^2{\vartheta}$.
\end{remark}
The differential equation (\ref{Smeno}) has two singularities at $\vartheta=0,\pi$. An additional singularity emerges if $\lambda$ belongs to the interval $(-am_e,am_e)$. In the case that $\lambda=\pm am_e$ the third singularity coincides with the singularities at $\pi$ or $0$, respectively. In what follows we will assume that such a singularity does not belong to the interval $(0,\pi)$. Let
\[
S(\vartheta)=\left(\begin{array}{c}
S_{+}(\vartheta)\\
S_{-}(\vartheta)
\end{array}
\right)=\frac{\widetilde{S}(\vartheta)}{\sqrt{\sin{\vartheta}}}=\frac{1}{\sqrt{\sin{\vartheta}}}\left(\begin{array}{c}
\widetilde{S}_{+}(\vartheta)\\
\widetilde{S}_{-}(\vartheta)
\end{array}
\right),
\]
and introduce the coordinate transformation $dx/d\vartheta=1/\sqrt{\Delta_\vartheta}$. A simple integration gives
\begin{equation}\label{xt}
x(\vartheta)=\frac{1}{\sqrt{\Xi}}\int_0^\vartheta\frac{d\widetilde{\vartheta}}{\sqrt{1-k^2\sin^2{\widetilde{\vartheta}}}},\quad 0<k^2=\frac{\sigma}{1+\sigma}<1,\quad\sigma=\frac{\Lambda}{3}a^2,
\end{equation}
where the integration constant has been chosen to be zero since $x=\vartheta$ in the limit $\sigma\to 0$. With the help of $110.02$, $110.06$, and $113.01$ in \cite{Byrd} we find that
\begin{equation}\label{solnx}
x(\vartheta)
= \left\{ \begin{array}{ll}
         x_1(\vartheta)=\frac{1}{\sqrt{\Xi}}F\left(\sin{\vartheta},\sqrt{\frac{\sigma}{1+\sigma}}\right) & \mbox{if $0<\vartheta\leq \pi/2$},\\
         x_2(\vartheta)=\frac{1}{\sqrt{\Xi}}K\left(\sqrt{\frac{\sigma}{1+\sigma}}\right)-F(\cos{\vartheta},i\sqrt{\sigma}) & \mbox{if $\pi/2<\vartheta\leq \pi$}.\end{array} \right. 
\end{equation}
Here, $K$ and $F$ denote the complete and incomplete elliptic integral of the first kind, respectively. Furthermore, using $110.06$ in \cite{Byrd} it is not difficult to verify that $x(0)=0$, 
\[
x_1(\pi/2)=x_2(\pi/2)=x_0,\quad
x_2(\pi)=x_0+K(i\sqrt{\sigma}),\quad 
x_0=\frac{1}{\sqrt{\Xi}}K\left(\sqrt{\frac{\sigma}{1+\sigma}}\right).
\]
This means that the interval $(0,\pi)$ will be mapped by the coordinate transformation $x=x(\vartheta)$ to the positive interval $I=(0,x_2(\pi))$. Observe that for $\sigma\to 0$ we can use $111.02$ in \cite{Byrd} to show that the interval $I$ reduces as expected to the interval $(0,\pi)$. Last but not least, we have
\[
\lim_{\sigma\to 0}x_1(\vartheta)=\vartheta=
\lim_{\sigma\to 0}x_2(\vartheta)
\]
signalizing that (\ref{solnx}) reduces correctly to $x=\vartheta$ in the limit $\sigma\to 0$. The coordinate transformation $x=x(\vartheta)$ can be inverted and expressed in terms of the Jakobi elliptic functions as follows
\begin{equation}\label{solntheta}
\vartheta(x)
= \left\{ \begin{array}{ll}
         \vartheta_1(x)=\sin^{-1}\left(\mbox{sn}\left(\sqrt{\Xi}x,\sqrt{\frac{\sigma}{1+\sigma}}\right)\right) & \mbox{if $0<x\leq x_0$},\\
         \vartheta_2(x)=\cos^{-1}\left(\mbox{sn}\left(\frac{K\left(\sqrt{\frac{\sigma}{1+\sigma}}\right)-\sqrt{\Xi}x}{\sqrt{\Xi}},i\sqrt{\sigma}\right)\right) & \mbox{if $x_0<x<x_2$}.\end{array} \right. 
\end{equation}
Finally, the angular system can be rewritten as a Dirac system
\begin{equation}\label{formal}
(\mathcal{U}\widetilde{S})(x)=\left(
\begin{array}{cc}
0&1\\
-1&0
\end{array}
\right)\frac{d\widetilde{S}}{dx}+\displaystyle{
\left(
\begin{array}{cc}
-am_e\cos{\vartheta(x)}&-\frac{\Xi}{\sqrt{\Delta_x}}\left[a\omega\sin{\vartheta(x)}+\frac{\widehat{k}}{\sin{\vartheta(x)}}\right]\\
-\frac{\Xi}{\sqrt{\Delta_x}}\left[a\omega\sin{\vartheta(x)}+\frac{\widehat{k}}{\sin{\vartheta(x)}}\right] &am_e\cos{\vartheta(x)}
\end{array}
\right)}\widetilde{S}(x)=\lambda\widetilde{S}(x)
\end{equation}
with $x\in I$ and $\theta(x)$ defined as in (\ref{solntheta}). Let us rewrite the formal differential operator $\mathcal{U}$  as
\begin{equation}\label{exprU}
\mathcal{U}=\left(
\begin{array}{cc}
-am_e\cos{\vartheta(x)}&\frac{d}{dx}-\frac{\Xi}{\sqrt{\Delta_x}}\left[a\omega\sin{\vartheta(x)}+\frac{\widehat{k}}{\sin{\vartheta(x)}}\right]\\
-\frac{d}{dx}-\frac{\Xi}{\sqrt{\Delta_x}}\left[a\omega\sin{\vartheta(x)}+\frac{\widehat{k}}{\sin{\vartheta(x)}}\right] &am_e\cos{\vartheta(x)}
\end{array}
\right)
\end{equation}
which acts on the Hilbert space $L^2(I,dx)^2$. To simplify the following analysis we express $\mathcal{U}$ as the sum of the unbounded operator $\mathcal{U}_u$ and the bounded operator $\mathcal{U}_b$ given by
\[
\mathcal{U}_u=\left(
\begin{array}{cc}
0&\frac{d}{dx}-\frac{\Xi\widehat{k}}{\sqrt{\Delta_x}\sin{\vartheta(x)}}\\
-\frac{d}{dx}-\frac{\Xi\widehat{k}}{\sqrt{\Delta_x}\sin{\vartheta(x)}}&0
\end{array}
\right),\quad
\mathcal{U}_b=\left(
\begin{array}{cc}
-am_e\cos{\vartheta(x)}&-\Xi a\omega\frac{\sin{\vartheta(x)}}{\sqrt{\Delta_x}}\\
-\Xi a\omega\frac{\sin{\vartheta(x)}}{\sqrt{\Delta_x}}&am_e\cos{\vartheta(x)}
\end{array}
\right).
\]
In the following we always assume that $\Lambda$, $a$, $m_e$, $\omega$ are real and $k\in\mathbb{Z}$. The minimal operator associated to $\mathcal{U}$ is $A^{min}_u\Phi=\mathcal{U}_u\Phi$ with domain of definition $D(A^{min}_u)=C_0^\infty(I)^2$. Let us introduce the inner product 
\[
\langle\Phi,\widetilde{\Phi}\rangle=\int_I dx~\Phi(x)\widetilde{\Phi}^{*}(x)
\]
where $*$ denotes complex conjugation followed by transposition. Then, it is not difficult to verify that $A^{min}_u$ is formally self-adjoint. Hence, by Theorem~$5.4$ in \cite{Weidmann1} the operator $A^{min}_u$ will be closable. Let $A_u$ denote the closure of $A^{min}_u$. We can prove the following
\begin{theorem}\label{THMI}
 The operator $A_u$ is self-adjoint if and only if $k\in\mathbb{R}\backslash(-1,0)$ and a fortiori for every $k\in\mathbb{Z}$.
\end{theorem}
\begin{proof}
We show that $\mathcal{U}_u$ is in the limit point case at $\vartheta=0$ and $\vartheta=\pi$. A fundamental system of the differential equation $\mathcal{U}_u\Phi=0$ is
\begin{eqnarray*}
\Phi_1(x)&=&\left(\tan{\frac{\vartheta(x)}{2}}\right)^{-\left(k+\frac{1}{2}\right)}e^{\sqrt{\sigma}\left(k+\frac{1}{2}\right)\tan^{-1}(\sqrt{\sigma}\cos{\vartheta(x)})}
\left(\begin{array}{c}
1\\
0
\end{array}\right),\\
\Phi_2(x)&=&\left(\tan{\frac{\vartheta(x)}{2}}\right)^{k+\frac{1}{2}}e^{-\sqrt{\sigma}\left(k+\frac{1}{2}\right)\tan^{-1}(\sqrt{\sigma}\cos{\vartheta(x)})}
\left(\begin{array}{c}
0\\
1
\end{array}\right).
\end{eqnarray*} 
Let us analyze the square integrability of these solutions. To this purpose let $k\geq 0$, $\alpha\in (x_0,x_2)$, and $\beta\in(0,x_0)$. Then, we have
\[
\int_\alpha^{x_2}dx~|\Phi_2(x)|^2>\int_\alpha^{x_2}dx~\left(\tan{\frac{\vartheta(x)}{2}}\right)^{2k+1}
>\int_\alpha^{x_2}dx~\tan{\frac{\vartheta(x)}{2}}=\int_{\vartheta(\alpha)}^{\pi}d\vartheta~\frac{\tan{(\vartheta/2)}}{\sqrt{\Delta_\vartheta}}
\]
\[
=\left.\ln{\frac{2\Xi-4\sigma\cos^2{(\vartheta/2)}+2\Xi\sqrt{\Xi-2\sigma\sin^2{\vartheta}}}{\cos^2{(\vartheta/2)}}}\right|^\pi_{\vartheta(\alpha)}=+\infty,
\]
where in the first majorization we used the fact that $\cos{\vartheta(x)}<0$ for $x\in(\alpha,x_2)$ together with $\tan(\vartheta(x)/2)>0$ since $\theta/2\in(\vartheta(\alpha),\pi/2)$ and therefore $\tan^{-1}{(\sqrt{\sigma}\cos{\vartheta(x)})}$ is negative. The second majorization is obtained by observing that $\vartheta(\alpha)/2>\pi/4$ implies that $\tan{(\vartheta(x)/2)}>0$. Furthermore, $\vartheta(x)/2>\pi/4$ and hence $\tan{(\vartheta(x)/2)}>1$ on the interval $L=(\vartheta(\alpha)/2,\pi/2)$. This implies that $(\tan{(\vartheta(x)/2)})^{2k+1}>\tan{(\vartheta(x)/2)}$ on $L$. Moreover, we also have
\[
\int_0^\beta dx~|\Phi_1(x)|^2>\int_0^\beta dx~\left(\cot{\frac{\vartheta(x)}{2}}\right)^{2k+1}
>\int_0^\beta dx~\cot{\frac{\vartheta(x)}{2}}=\int_0^{\vartheta(\beta)} d\vartheta~\frac{\cot{(\vartheta/2)}}{\sqrt{\Delta_\vartheta}}>
\]
\[
\frac{1}{\sqrt{\Xi}}\int_0^{\vartheta(\beta)} d\vartheta~\cot{(\vartheta/2)}
=\left.\ln{\sin{\frac{\vartheta}{2}}}\right|_0^{\vartheta(\beta)}=+\infty,
\]
where we used the fact that for $\vartheta\in(0,\beta)$ we have $\sqrt{\Delta_\vartheta}<\sqrt{\Xi}$. On the other hand, we also have the following estimates
\[
\int^{x_2}_\beta dx~|\Phi_1(x)|^2<e^{\frac{\pi}{2}\sqrt{\sigma}(2k+1)}\int^{x_2}_\beta dx~\left(\cot{\frac{\vartheta(x)}{2}}\right)^{2k+1}
\leq e^{\frac{\pi}{2}\sqrt{\sigma}(2k+1)}\left(\cot{\frac{\vartheta(\beta)}{2}}\right)^{2k+1}\int^{x_2}_\beta dx<+\infty
\]
and
\[
\int^{\alpha}_0 dx~|\Phi_2(x)|^2<e^{\frac{\pi}{2}\sqrt{\sigma}(2k+1)}\int^{\alpha}_0 dx~\left(\tan{\frac{\vartheta(x)}{2}}\right)^{2k+1}
\leq e^{\frac{\pi}{2}\sqrt{\sigma}(2k+1)}\left(\tan{\frac{\vartheta(\alpha)}{2}}\right)^{2k+1}\int^{\alpha}_0 dx<+\infty.
\]
This shows that in the case $k\geq 0$ the solution $\Phi_1$ lies right but not left in $L^2(I,dx)^2$, whereas the solution $\Phi_2$ lies left in $L^2(I,dx)^2$ but it does not lie right in $L^2(I,dx)^2$. For $k\leq -1$ the same holds true for $\Phi_1$ and $\Phi_2$ exchanged. According to Weyl's alternative it follows that for $k\in\mathbb{R}\backslash(-1,0)$ the formal differential operator $\mathcal{U}_u$ is in the limit point case both at $0$ and at $\pi$. Hence, Theorem~2.7 in \cite{Weidmann2} ensures that the closure of $A^{min}_u$ is self-adjoint. To show that $A_u$ is not self-adjoint for $k\in(-1,0)$ we check that $\Phi_1$ and $\Phi_2$ are in $L^2(I,dx)^2$, thus $\mathcal{U}_u$ is in the limit circle case both at $0$ and $\pi$. Then, again by Theorem~2.7 in \cite{Weidmann2} the assertion follows. We give a proof for $\Phi_2\in L^2(I,dx)^2$ in the case $k\in(-1,1/2]$, the remaining cases can be treated similarly. From the initial assumption we have $2k+1\in(-1,0]$. Hence, it follows from the inequality $\sin{\frac{\gamma}{2}}\geq\frac{\gamma}{\pi}$ with $\gamma\in(0,\pi)$ and the monotonicity of the cosine and tangent functions
\[
\int^{x_2}_0 dx~|\Phi_2(x)|^2=\int^{x_2}_0~dx|\Phi_2(x)|^2=
e^{\frac{\pi}{2}\sqrt{\sigma}|2k+1|}\int^{x_2}_0~dx\left(\tan{\frac{\vartheta(x)}{2}}\right)^{2k+1}=
e^{\frac{\pi}{2}\sqrt{\sigma}|2k+1|}\int_{0}^{\pi}~\frac{d\vartheta}{\sqrt{\Delta_\vartheta}}~\left(\tan{\frac{\vartheta}{2}}\right)^{2k+1}\leq
\]
\[
e^{\frac{\pi}{2}\sqrt{\sigma}|2k+1|}\int_{0}^{\pi}~d\vartheta\left(\tan{\frac{\vartheta}{2}}\right)^{2k+1}=
e^{\frac{\pi}{2}\sqrt{\sigma}|2k+1|}\left[\int^{\pi/2}_0\left(\cot{\frac{\vartheta}{2}}\right)^{|2k+1|}+
\int_{\pi/2}^{\pi}~d\vartheta\left(\cot{\frac{\vartheta}{2}}\right)^{|2k+1|}\right]\leq
\]
\[
e^{\frac{\pi}{2}\sqrt{\sigma}|2k+1|}\left[\int^{\pi/2}_0\left(\cot{\frac{\vartheta}{2}}\right)^{|2k+1|}+
\int_{\pi/2}^{\pi}~d\vartheta\left(\cot{\frac{\vartheta}{2}}\right)^{|2k+1|}\right]\leq
e^{\frac{\pi}{2}\sqrt{\sigma}|2k+1|}\left[\int^{\pi/2}_0\left(\sin{\frac{\vartheta}{2}}\right)^{-|2k+1|}+
+\frac{\pi}{2}\right]\leq
\]
\[
e^{\frac{\pi}{2}\sqrt{\sigma}|2k+1|}\left[\pi^{|2k+1|}\int^{\pi/2}_0 d\vartheta\vartheta^{-|2k+1|}+\frac{\pi}{2}\right]
=\frac{\pi}{2}e^{\frac{\pi}{2}\sqrt{\sigma}|2k+1|}\left(1+\frac{2^{|2k+1|}}{1-|2k+1|}\right)<+\infty.
\]
This completes the proof.~~$\square$
\end{proof}
In order to find an explicit representation for the domain of the operator $A_u$ we introduce the so-called maximal operator associated to $\mathcal{U}_u$ by $D(A^{max}_u)=\{\Phi\in L^2(I,dx)^2~|~\Phi~\mbox{is absolutely continuous},~\mathcal{U}_u\Phi\in L^2(I,dx),~A^{max}_u\Phi:=\mathcal{U}_u\Phi\}$. Then, by Theorem~3.9 in \cite{Weidmann2} it follows that the adjoint $A^{*}_u=A^{max}_u$. Since for $k\in\mathbb{R}\backslash(-1,0)$ the operator $A_u$ is self-adjoint by Theorem~\ref{THMI}, we also have that $A_u=A_u^{max}$.
\begin{theorem}\label{IV3}
The angular operator $A\widetilde{S}:=\mathcal{U}\widetilde{S}$ with $\mathcal{U}$ given by (\ref{exprU}) and domain of definition
\[
D(A)=\{\widetilde{S}\in L^2(I,dx)^2~|~\widetilde{S}~\mbox{is absolutely continuous},~\mathcal{U}\widetilde{S}\in L^2(I,dx)^2\}
\]
is self-adjoint if and only if $k\in\mathbb{R}\backslash(-1,0)$. In this case, $A$ is the closure of the minimal operator $A^{min}$ defined by $D(A^{min})=C_0^\infty(I)^2$ and $A^{min}\widetilde{S}:=\mathcal{U}\widetilde{S}$.
\end{theorem}
\begin{proof}
Let us start by observing that $D(A)=D(A^{max}_u)$. Let $A_b$ be the maximal operator associated with the formal multiplication operator $\mathcal{U}_b$, that is $D(A_b)=L^2(I,dx)^2$ and $A_b\Psi=\mathcal{U}_b\Psi$. The operator $A_b$ is symmetric and bounded in the Hilbert space $L^2(I,dx)^2$. Hence, Theorem~$4.10$ (see Ch. V in \cite{Kato}) shows that $A=A_u+A_b$ with domain $D(A)=D(A_u)$ is self-adjoint if and only if $A_u$ is self-adjoint. The result follows from Theorem~\ref{THMI}.~~$\square$
\end{proof}
Since $A$ is self-adjoint its spectrum $\sigma(A)$ must be real. 
\begin{remark}
Consider now the angular operator in the special case $m_e=0$ and introduce the formal differential expressions
\[
\mathfrak{B}=\left(
\begin{array}{cc}
0 & \mathfrak{B}_{+}\\
\mathfrak{B}_{-} & 0
\end{array}
\right)=\left(
\begin{array}{cc}
0 & \frac{d}{dx}-\frac{\Xi}{\sqrt{\Delta_x}}\left[a\omega\sin{\vartheta(x)}+\frac{\widehat{k}}{\sin{\vartheta(x)}}\right]\\
-\frac{d}{dx}-\left[a\omega\sin{\vartheta(x)}+\frac{\widehat{k}}{\sin{\vartheta(x)}}\right] & 0
\end{array}
\right).
\]
Then, Theorem~\ref{IV3} implies that for $k\in\mathbb{R}\backslash(-1,0)$ the operator $\mathcal{B}\widetilde{S}=\mathfrak{B}\widetilde{S}$ with domain of definition $D(\mathcal{B})=\{\widetilde{S}\in L^2(I,dx)^2~|~\widetilde{S}~\mbox{absolutely continuous},~\mathfrak{B}\widetilde{S}\in L^2(I,dx)^2\}$ is self-adjoint and it is the closure of the minimal operator $\mathcal{B}^{min}$ given by $D(\mathcal{B}^{min})=C_0^\infty(I)^2$ with $\mathcal{B}^{min}\widetilde{S}=\mathfrak{B}\widetilde{S}$. This implies that the operators $B\widetilde{S}_2=\mathfrak{B}_{+}\widetilde{S}_2$ with $D(B)=\{\widetilde{S}_2\in L^2(I,dx)~|~\widetilde{S}_2~\mbox{absolutely continuous},~\mathfrak{B}_{+}\widetilde{S}_2\in L^2(I,dx)\}$ and $B_{-}\widetilde{S}_1=\mathfrak{B}_{-}\widetilde{S}_1$ with $D(B_{-})=\{\widetilde{S}_1\in L^2(I,dx)~|~\widetilde{S}_1~\mbox{absolutely continuous},~\mathfrak{B}_{-}\widetilde{S}_1\in L^2(I,dx)\}$ are adjoint to each other so that $\mathcal{B}=\begin{pmatrix}
0 & B\\
B^{*} &0 \\
\end{pmatrix}$. 
\end{remark}
Let us write the angular operator $A$ as $A=\begin{pmatrix}
D & B\\
B^{*} &D \\
\end{pmatrix}$ with $D=am_e\cos{\vartheta(x)}$ a bounded multiplication operator in $L^2(I,dx)^2$ and 
\[
B=\frac{d}{dx}-\frac{\Xi}{\sqrt{\Delta_x}}\left[a\omega\sin{\vartheta(x)}+\frac{\widehat{k}}{\sin{\vartheta(x)}}\right].
\]
By using an off-diagonalization method as in \cite{Wink} we show that $A$ has compact resolvent. This together with Theorem~$6.29$, Ch.III in \cite{Kato} will imply that the spectrum of $A$ consists only of isolated eigenvalues with no accumulation points $(-\infty,\infty)$. First, we prove that the discrete spectrum of $B$ and $B^{*}$ is empty.
\begin{lemma}\label{319}
$\sigma_p(B)=\sigma_p(B^{*})=\emptyset$.
\end{lemma}
\begin{proof}
Take any $\mu\in\mathbb{C}$. Then, $\mu\in\sigma_p(B)\cup\sigma_p(B^{*})$ if and only if at least one of the differential equations
\[
\left(-\frac{d}{dx}-\frac{\Xi\widehat{k}}{\sqrt{\Delta_x}\sin{\vartheta(x)}}-\frac{\Xi a\omega\sin{\vartheta(x)}}{\sqrt{\Delta_x}}-\mu\right)\varphi_{[\mu]}(x)=0,\quad
\left(\frac{d}{dx}-\frac{\Xi\widehat{k}}{\sqrt{\Delta_x}\sin{\vartheta(x)}}-\frac{\Xi a\omega\sin{\vartheta(x)}}{\sqrt{\Delta_x}}-\mu\right)\psi_{[\mu]}(x)=0
\]
has a square integrable solution. Let $\Xi=1+\sigma$ and recall that $\widehat{k}=k+1/2$. The solutions of these differential equations are
\[
\varphi_{[\mu]}(x)=ce^{f(x)}\left(\tan{\frac{\vartheta(x)}{2}}\right)^{-\left(k+\frac{1}{2}\right)},\quad
\psi_{[\mu]}(x)=ce^{-f(x)}\left(\tan{\frac{\vartheta(x)}{2}}\right)^{k+\frac{1}{2}}
\]
with
\[
f(x)=-\mu x+\left(k+\frac{1}{2}\right)\sqrt{\sigma}\arctan{(\sqrt{\sigma}\cos{\vartheta(x)})}+a\omega(1+\sigma)\frac{\arctan{(\sqrt{\sigma}\cos{\vartheta(x)})}}{\sqrt{\sigma}}.
\]
The functions $\varphi_{[\mu]}$ and $\psi_{[\mu]}$ are defined up to a multiplicative constant $c\in\mathbb{C}$. Without loss of generality we set $c=1$. Let us show that this functions are not square integrable on the interval $I=(0,x_2)$. We start by observing that
\[
\langle\varphi_{[\mu]},\varphi_{[\mu]}\rangle=\int_I\varphi_{[\mu]}(x)\overline{\varphi_{[\mu]}(x)}~dx=
\int_0^\pi\varphi_{[\mu]}(\vartheta)\overline{\varphi_{[\mu]}(\vartheta)}~\frac{d\vartheta}{\sqrt{\Delta_\vartheta}}=
\int_0^\pi e^{\rho(\vartheta)}\left(\tan{\frac{\vartheta}{2}}\right)^{-(2k+1)}~\frac{d\vartheta}{\sqrt{\Delta_\vartheta}}
\]
with
\[
\rho(\vartheta)=-2(\Re{\mu})\vartheta+(2k+1)\sqrt{\sigma}\arctan{(\sqrt{\sigma}\cos{\vartheta})}+2a\omega(1+\sigma)\frac{\arctan{(\sqrt{\sigma}\cos{\vartheta})}}{\sqrt{\sigma}}.
\]
Taking into account that $1\leq\sqrt{\Delta_\vartheta}\leq\sqrt{1+\sigma}$ for $\vartheta\in[0,\pi]$ and that the function $\arctan{(\sqrt{\sigma}\cos{\vartheta})}$ is continuous and decreasing on $[0,\pi]$ let $M=\inf_{\vartheta\in[0,\pi]}\{e^{\rho(\vartheta)}\}>0$. Then, 
\[
\langle\varphi_{[\mu]},\varphi_{[\mu]}\rangle>\frac{M}{\sqrt{1+\sigma}}\left[\int_0^{\pi/2}\left(\tan{\frac{\vartheta}{2}}\right)^{-(2k+1)}~d\vartheta+\int_{\pi/2}^{\pi}\left(\tan{\frac{\vartheta}{2}}\right)^{-(2k+1)}~d\vartheta\right].
\]
For $k\geq 0$ we have $\tan{(\vartheta/2)}\leq 1$ when $\vartheta\in(0,\pi/2)$ and therefore
\[
\int_0^{\pi/2}\left(\tan{\frac{\vartheta}{2}}\right)^{-(2k+1)}~d\vartheta\geq
\int_0^{\pi/2}\left(\tan{\frac{\vartheta}{2}}\right)^{-1}~d\vartheta=\left.2\ln{\sin{\frac{\vartheta}{2}}}\right|_0^{\pi/2}=+\infty.
\]
When $k\leq -1$ we have $\tan{(\vartheta/2)}\geq 1$ on $(\pi/2,\pi)$ and hence
\[
\int_{\pi/2}^{\pi}\left(\tan{\frac{\vartheta}{2}}\right)^{-(2k+1)}~d\vartheta\geq
\int_{\pi/2}^\pi\tan{\frac{\vartheta}{2}}~d\vartheta=\left.-2\ln{\cos{\frac{\vartheta}{2}}}\right|_{\pi/2}^\pi=+\infty.
\]
In any case it results $\varphi_{[\mu]}\notin L^2(I,dx)$. Similarly, it can be shown that $\psi_{[\mu]}\notin L^2(I,dx)$.~~$\square$
\end{proof}
For $\mu\in\mathbb{C}$ we introduce the formal differential expression defined by $\mathfrak{B}_\mu=\begin{pmatrix}
0 & \mathfrak{B}_{+}-\mu\\
\mathfrak{B}_{-}-\overline{\mu} &0 \\
\end{pmatrix}$ and we associate to $\mathfrak{B}_\mu$ the differential operator $\mathcal{B}_\mu\widetilde{S}=\mathfrak{B}_\mu\widetilde{S}$ with $D(\mathcal{B}_\mu)=D(A)$. Moreover, in the notation of the previous remark we have $\mathcal{B}=\mathcal{B}_0=\begin{pmatrix}
0 & B\\
B^{*} &0 \\
\end{pmatrix}$. The operator $\mathcal{B}_\mu$ is self-adjoint for any $\mu\in\mathbb{C}$ because $A$ is self-adjoint and $\mathcal{B}_\mu-A$ is symmetric and bounded. Following an analogous proof to that of Theorem~$2.14$ in \cite{Wink} it can be shown that $\sigma_{ess}(\mathcal{B}_\mu)=\emptyset$. This fact together with the next result implies that $\mathcal{B}_\mu$ is boundedly invertible.
\begin{lemma}\label{pl}
For all $\mu\in\mathbb{C}$ we have $0\notin\sigma_p(\mathcal{B}_\mu)$.
\end{lemma}
\begin{proof}
We prove it by contradiction. Suppose $0\in\sigma_p(\mathcal{B}_\mu)$ and $\widetilde{S}$ be an eigenfunction of $\mathcal{B}_\mu$ with  eigenvalue $0$. Then,
\[
0=\mathcal{B}_\mu\widetilde{S}=\mathcal{B}_\mu\begin{pmatrix}
0 & I\\
I &0 \\
\end{pmatrix}\begin{pmatrix}
0 & I\\
I &0 \\
\end{pmatrix}\widetilde{S}=\begin{pmatrix}
B-\mu & 0\\
0 &B^{*}-\overline{\mu} \\
\end{pmatrix}\begin{pmatrix}
0 & I\\
I &0 \\
\end{pmatrix}\widetilde{S}
\]
implies that either $B-\mu$ or $B^{*}-\overline{\mu}$ is not injective which is in contradiction to $\sigma_p(B)\cup\sigma_p(B^{*})=\emptyset$ following from the previous lemma.~~$\square$
\end{proof}
We now derive an auxiliary result that will allow us to prove that the operator $\mathcal{B}$ has compact resolvent from which it will follow that the angular operator has compact resolvent  as well. According to Lemma~\ref{pl} we have $0\in\mathbb{C}\backslash(\sigma_p(\mathcal{B}_\mu)\cup\sigma_{ess}(\mathcal{B}_\mu))=\rho(\mathcal{B}_\mu)$, where $\rho(\mathcal{B}_\mu)$ denotes the the resolvent set of $\mathcal{B}_\mu$. Hence, $B-\mu$ and $B^{*}-\overline{\mu}$ are boundedly invertible. Moreover, their resolvents and the resolvent of $\mathcal{B}_\mu$ are connected as follows
\[
\begin{pmatrix}
(B-\mu)^{-1} & 0\\
0 & (B^{*}-\overline{\mu})^{-1} \\
\end{pmatrix}=\begin{pmatrix}
B-\mu & 0\\
0 & B^{*}-\overline{\mu} \\
\end{pmatrix}^{-1}=\begin{pmatrix}
0 & I\\
I &0 \\
\end{pmatrix}\mathcal{B}_\mu^{-1}.
\]
In particular, we have that the ranges of $B-\mu$ and $B^{*}-\overline{\mu}$ are such that $\mbox{rg}(B-\mu)=\mbox{rg}(B^{*}-\overline{\mu})=L^2(I,dx)$. Hence, we have shown that $\sigma(B)=\sigma(B^{*})=\emptyset$.
\begin{lemma}\label{321}
Let $\mu\in\mathbb{C}$ and $\varphi_{[\mu]}$ and $\psi_{[\mu]}$ be defined as in the proof of Lemma~\ref{319}. Then, the operators $(B-\mu)^{-1}$ and $(B^{*}-\overline{\mu})^{-1}$ map functions $g,h\in L^2(I,dx)$ to
\begin{equation}\label{343a}
(B-\mu)^{-1}g(x)=\frac{1}{\varphi_{[\mu]}(x)}\left\{ \begin{array}{ll}
         \int_x^{x_2}\varphi_{[\mu]}(t)g(t)~dt & \mbox{if}~ k \geq 0\vspace{0.2cm}\\
         \int_{x}^0\varphi_{[\mu]}(t)g(t)~dt & \mbox{if}~ k \leq -1\end{array}\right.,\quad x\in I
\end{equation}
and
\begin{equation}\label{343b}
(B^{*}-\overline{\mu})^{-1}h(x)=\frac{1}{\psi_{[\mu]}(x)}\left\{ \begin{array}{ll}
         \int_0^{x}\psi_{[\mu]}(t)h(t)~dt & \mbox{if}~ k \geq 0\vspace{0.2cm}\\
         \int_{x_2}^x\psi_{[\mu]}(t)h(t)~dt & \mbox{if}~ k \leq -1\end{array}\right.,\quad x\in I.
\end{equation}
\end{lemma}
\begin{proof}
From the proof Lemma~\ref{319} it follows that $\psi_{[\mu]}$ is a solution of $(B-\mu)u=0$ and that $\varphi_{[\mu]}$ is a solution of $(B^{*}-\overline{\mu})v=0$. To prove (\ref{343a}) and (\ref{343b}) we first show that $(B-\mu)G(x)=g(x)$ and $(B^{*}-\overline{\mu})H(x)=h(x)$ with $x\in I$ hold formally, where $G$ and $H$ denote the r.h.s. of (\ref{343a}) and (\ref{343b}), respectively. Let us start by assuming $k\geq 0$. Then, for $g\in\mbox{rg}(B-\mu)$ we find
\begin{eqnarray*}
(B-\mu)G(x)&=&(B-\mu)\left(\frac{1}{\varphi_{[\mu]}(x)}\int_x^{x_2}\varphi_{[\mu]}(t)g(t)~dt\right)
=(B-\mu)\left(\psi_{[\mu]}(x)\int_x^{x_2}\varphi_{[\mu]}(t)g(t)~dt\right),\\
&=&\left(B\psi_{[\mu]}(x)-\mu\psi_{[\mu]}(x)\right)\int_0^x\varphi_{[\mu]}(t)g(t)~dt+\psi_{[\mu]}(x)\frac{d}{dx}\int_0^x\varphi_{[\mu]}(t)g(t)~dt=\psi_{[\mu]}(x)\varphi_{[\mu]}(x)g(x)=g(x),
\end{eqnarray*}
where we used $\psi_{[\mu]}(x)\varphi_{[\mu]}(x)=1$ and $(\mathfrak{B}_{+}-\mu)\psi_{[\mu]}=0$. The case $k\leq -1$ and the equation for $h$ can be shown in a similar way. Let us prove that $G\in D(B)$ and $H\in D(B^{*})$. We will do it only for $G\in D(B)$ in the case $k\geq 0$ because the statement for $k\leq -1$ and $H\in D(B^{*})$ can be proved analogously. Since $D(B)=\{g\in L^2(I,dx)~|~g~\mbox{absolutely continuous},~\mathfrak{B}_{+}g\in L^2(I,dx)\}$, we need to verify that $G\in L^2(I,dx)$. Let $k\geq 0$. Then,
\begin{eqnarray}
\|G\|^2_{L^2}&=&\int_0^{x_2}|G(x)|^2~dx=\int_0^{x_2}\frac{1}{|\varphi_{[\mu]}(x)|^2}\left|\int_x^{x_2}\varphi_{[\mu]}(t)g(t)~dt\right|^2 dx\leq
\int_0^{x_2}\left(\int_x^{x_2}\left|\frac{\varphi_{[\mu]}(t)}{\varphi_{[\mu]}(x)}\right||g(t)| dt\right)^2 dx,\notag\\
&=&\int_0^{\pi}\left(\int_\vartheta^\pi\left|\frac{\varphi_{[\mu]}(\tau)}{\varphi_{[\mu]}(\vartheta)}\right||g(\tau)|\frac{d\tau}{\sqrt{\Delta_{\tau}}}\right)^2\frac{d\vartheta}{\sqrt{\Delta_\vartheta}}<\int_0^{\pi}\left(\int_\vartheta^\pi\left|\frac{\varphi_{[\mu]}(\tau)}{\varphi_{[\mu]}(\vartheta)}\right||g(\tau)| d\tau\right)^2 d\vartheta\label{344}
\end{eqnarray}
with $0<\vartheta<\tau<\pi$. Note that by assumption $g\in L^2((0,\pi),d\vartheta/\sqrt{\Delta_\vartheta})$ and its restriction on the interval $(0,\vartheta)$ will belong to $L^2((0,\vartheta),d\tau/\sqrt{\Delta_\tau})$ for $\vartheta\in(0,\pi)$. Furthermore, we have
\[
\left|\frac{\varphi_{[\mu]}(\tau)}{\varphi_{[\mu]}(\vartheta)}\right|^2=e^{-2(\Re{\mu})(\tau-\vartheta)+\chi(\tau,\vartheta)}
\left(\frac{\tan{\frac{\vartheta}{2}}}{\tan{\frac{\tau}{2}}}\right)^{2k+1}
\]
with
\[
\chi(\tau,\vartheta)=2\left[\widehat{k}\sqrt{\sigma}+a\omega\frac{1+\sigma}{\sqrt{\sigma}}\right]
\left[\arctan{(\sqrt{\sigma}\cos{\tau})}-\arctan{(\sqrt{\sigma}\cos{\vartheta})}\right].
\]
First of all observe that $e^{-2(\Re{\mu})(\tau-\vartheta)}=e^{2(\Re{\mu})\vartheta}e^{-2(\Re{\mu})\tau}\leq Ke^{2|\Re{\mu}|\vartheta}$ with 
$K=\max{\{1,e^{2|\Re{\mu}|\pi}\}}$. Furthermore, the function $\chi(\tau,\vartheta)$ has an absolute maximum on the square $[0,\pi]\times[0,\pi]$ at the point $(0,\pi)$ where $\chi(0,\pi)=2\arctan{\sqrt{\sigma}}$. Hence, there exists a positive constant $C$ such that $e^{\chi(\tau,\vartheta)}\leq C$ with $C=e^{4\left(|\widehat{k}|\sqrt{\sigma}+|a\omega|\frac{1+\sigma}{\sqrt{\sigma}}\right)\arctan{\sqrt{\sigma}}}$. Finally, since $0<\vartheta<\tau<\pi$ we can show that 
\[
\frac{\tan{\frac{\vartheta}{2}}}{\tan{\frac{\tau}{2}}}<\frac{\vartheta}{\tau}
\]
as follows. Since the above inequality is equivalent to the inequality $\frac{\tan{\frac{\vartheta}{2}}}{\vartheta}<\frac{\tan{\frac{\tau}{2}}}{\tau}$, we consider the function $f:(0,\pi)\longrightarrow\mathbb{R}$ such that $f(\vartheta)=\frac{\tan{\frac{\vartheta}{2}}}{\vartheta}$. Clearly, $f$ is continuously differentiable and the inequality we need to prove is equivalent to $f(\vartheta)<f(\tau)$ for $0<\vartheta<\tau<\pi$. Hence, it suffices to show that $f$ is a monotonously increasing function of $\vartheta$. A straightforward computation shows that
\[
f^{'}(\vartheta)=\frac{1}{2\vartheta\cos^2{\frac{\vartheta}{2}}}\left(1-\frac{\sin{\frac{\vartheta}{2}}}{\frac{\vartheta}{2}}\cos{\frac{\vartheta}{2}}\right)>0
\]
since $(\sin{\frac{\vartheta}{2}})/(\vartheta/2)<1$ and $\cos{\frac{\vartheta}{2}}<1$. Finally, $\vartheta<\tau$ implies that 
\[
\frac{\tan{\frac{\vartheta}{2}}}{\tan{\frac{\tau}{2}}}<1
\]
and we conclude that
\begin{equation}\label{ineq}
\left|\frac{\varphi_{[\mu]}(\tau)}{\varphi_{[\mu]}(\vartheta)}\right|^2<CK e^{2|\Re{\mu}|\vartheta}.
\end{equation}
Hence, $\varphi_{[\mu]}(\tau)/\varphi_{[\mu]}(\vartheta)\in L^2((0,\pi),d\tau)$ for each fixed $\vartheta\in(0,\pi)$. Therefore, we can use the Cauchy-Schwarz inequality to estimate the inner integral in (\ref{344}) as follows
\[
\left(\int_\vartheta^\pi\left|\frac{\varphi_{[\mu]}(\tau)}{\varphi_{[\mu]}(\vartheta)}\right||g(\tau)| d\tau\right)^2\leq
\int_\vartheta^\pi\left|\frac{\varphi_{[\mu]}(\tau)}{\varphi_{[\mu]}(\vartheta)}\right|^2 d\tau\int_\vartheta^\pi|g(\tau)|^2 d\tau
<CK\int_\vartheta^\pi e^{2|\Re{\mu}|\vartheta} d\tau \int_0^\pi|g(\tau)|^2 d\tau<C_0\|g\|^2_{L^2} e^{2|\Re{\mu}|\pi}(\pi-\vartheta)
\]
with $C_0=CK$. Inserting the above expression into (\ref{344}) yields $\|G\|^2_{L^2}<(\pi^2 C_0/2)e^{2|\Re{\mu}|\pi}\|g\|^2_{L^2}<+\infty$ and this concludes the proof.~~$\square$
\end{proof}
Note that Lemma~\ref{321} gives explicit expressions for the inverses of $B$ and $B^{*}$ by choosing $\mu=0$. Now that we have obtained an explicit form of $\mathcal{B}^{-1}$ we can show that $\mathcal{B}$, and therefore the angular operator $A$, has compact resolvent. 
\begin{lemma}\label{322}
The operator $\mathcal{B}$ has compact resolvent.
\end{lemma}
\begin{proof}
Tho show that the operator $\mathcal{B}^{-1}=\begin{pmatrix}
0 & (B^{*})^{-1}\\
B^{-1} & 0
\end{pmatrix}$ is compact, it suffices to show that the operators $B^{-1}$ and $(B^{*})^{-1}$ are compact. We limit us to prove that $B^{-1}$ is compact in the case $k\geq 0$ because the case $k\leq -1$ and the corresponding assertion regarding $(B^{*})^{-1}$ can be proved in an analogous way. From the previous lemma we have for $k\geq 0$ and $g\in L^2(I,dx)$
\[
B^{-1}g(x)=\frac{1}{\varphi_{[0]}(x)}\int_x^{x_2}\varphi_{[0]}(t)g(t) dt,\quad x\in I.
\]
For each $n\in\mathbb{N}$ we define the operators 
\[
T_n:L^2(I,dx)\longrightarrow L^2(I,dx),\quad T_n f(x)=\left\{ \begin{array}{ll}
         0 & \mbox{if}~ x\notin\left[\frac{1}{n},x_2-\frac{1}{n}\right]\vspace{0.2cm}\\
         \frac{1}{\varphi_{[0]}(x)}\int_{x+\frac{1}{n}}^{x_2}\varphi_{[0]}(t)f(t)~dt & \mbox{if}~ x\in\left[\frac{1}{n},x_2-\frac{1}{n}\right]\end{array}\right.
\]
and
\[
\widehat{T}_n:L^2\left(\left[\frac{1}{n},x_2-\frac{1}{n}\right],dx\right)\longrightarrow L^2\left(\left[\frac{1}{n},x_2-\frac{1}{n}\right],dx\right),\quad \widehat{T}_n f(x)=\frac{1}{\varphi_{[0]}(x)}\int_{x+\frac{1}{n}}^{x_2}\varphi_{[0]}(t)f(t)~dt.
\]
These operators are bounded for all $n\in\mathbb{N}$ and the operators $\widehat{T}_n$ are compact since the integral kernel is continuous and bounded (see example $4.1$, Ch. III in \cite{Kato}). For every $f\in L^2(I,dx)$ the restriction $\widehat{f}$ of $f$ on the compact interval $\left[\frac{1}{n},x_2-\frac{1}{n}\right]$ lies in $L^2\left(\left[\frac{1}{n},x_2-\frac{1}{n}\right],dx\right)$. It is clear that for any convergent sequence $(g_n)_{n\in\mathbb{N}}$ with terms in $L^2\left(\left[\frac{1}{n},x_2-\frac{1}{n}\right],dx\right)$ also the sequence $(\widetilde{g}_n)_{n\in\mathbb{N}}$ with terms in $L^2\left(\left[\frac{1}{n},x_2-\frac{1}{n}\right],dx\right)$ converges where
\[
\widetilde{g}_n(x)=\left\{ \begin{array}{ll}
         0 & \mbox{if}~ x\notin\left[\frac{1}{n},x_2-\frac{1}{n}\right]\vspace{0.2cm}\\
         g(x) & \mbox{if}~ x\in\left[\frac{1}{n},x_2-\frac{1}{n}\right]\end{array}\right.
\]
for all $n\in\mathbb{N}$. Let $(f_m)_{m\in\mathbb{N}}$ be a bounded sequence in $L^2(I,dx)$. Then, $(\widehat{f}_m)_{m\in\mathbb{N}}$ is also a bounded sequence in $L^2\left(\left[\frac{1}{n},x_2-\frac{1}{n}\right],dx\right)$. Hence, for any $n\in\mathbb{N}$ the sequence $(\widehat{T}_n f_m)_{m\in\mathbb{N}}$ contains a convergent subsequence. Therefore, also $(T_n f_m)_{m\in\mathbb{N}}$ contains a convergent subsequence since $\widehat{T}_n \widehat{f}_m=T_n f_m$. This shows that the operators $T_n$ are also compact. The proof is completed once we have shown that $T_n\to B^{-1}$ as $n\to\infty$ in the operator norm, that is we have to verify that $\|T_n-B^{-1}\|\to 0$ as $n\to\infty$. To see that, we note that for all $f\in L^2(I,dx)$ we have
\[
\|(T_n-B^{-1})f\|^2_{L^2}=\int_0^{x_2}|(T_n-B^{-1})f(x)|^2 dx=\int_0^{\frac{1}{n}}|B^{-1}f(x)|^2 dx+
\]
\[
+\int_{\frac{1}{n}}^{x_2-\frac{1}{n}}|T_n f(x)-B^{-1}f(x)|^2 dx+\int_{x_2-\frac{1}{n}}^{x_2}|B^{-1}f(x)|^2 dx=
\int_0^{\frac{1}{n}}|B^{-1}f(x)|^2 dx+\int_{x_2-\frac{1}{n}}^{x_2}|B^{-1}f(x)|^2 dx
\]
using the definition of $T_n$ and $B^{-1}$. Taking into account that for all $(a,b)\subseteq(0,x_2)$ we have
\[
\int_a^b|B^{-1}f(x)|^2 dx=\int_a^b\left|\frac{1}{\varphi_{[0]}(x)}\int_x^{x_2}\varphi_{[0]}(t)f(t) dt\right|^2 dx\leq
\int_a^b\left(\int_x^{x_2}\frac{\varphi_{[0]}(t)}{\varphi_{[0]}(x)}|f(t)| dt\right)^2 dx
\]
\[
<\int_{\vartheta(a)}^{\vartheta(b)}\left(\int_\vartheta^{\pi}\frac{\varphi_{[0]}(\tau)}{\varphi_{[0]}(\vartheta)}|f(\tau)| d\tau\right)^2 d\vartheta,
\]
where we used the fact that $\varphi_{[0]}$ is real-valued. Employing the Cauchy-Schwarz inequality and (\ref{ineq}) yields
\[
\int_a^b|B^{-1}f(x)|^2 dx<\int_{\vartheta(a)}^{\vartheta(b)}\left(\int_\vartheta^{\pi}\left(\frac{\varphi_{[0]}(\tau)}{\varphi_{[0]}(\vartheta)}\right)^2 d\tau\right)\left(\int_\vartheta^{\pi}|f(\tau)|^2 d\tau\right)d\vartheta\leq CK\int_{\vartheta(a)}^{\vartheta(b)}\left(\int_\vartheta^{\pi}1~ d\tau\right)\left(\int_\vartheta^{\pi}|f(\tau)|^2 d\tau\right)d\vartheta
\]
\[
\leq CK\int_{\vartheta(a)}^{\vartheta(b)}\left(\int_0^{\pi} d\tau\right)\left(\int_0^{\pi}|f(\tau)|^2 d\tau\right)d\vartheta=
\pi CK(\vartheta(b)-\vartheta(a))\|f\|^2_{L^2}.
\]
Hence,
\[
\|(T_n-B^{-1})f\|^2_{L^2}\leq\pi CK\left[\vartheta\left(\frac{1}{n}\right)+\vartheta(x_2)-\vartheta\left(x_2-\frac{1}{n}\right)\right]\|f\|^2_{L^2}
\]
where we used the fact that $\vartheta(0)=0$. Finally,
\[
\|T_n-B^{-1}\|=\leq \sqrt{\pi CK\left[\vartheta\left(\frac{1}{n}\right)+\vartheta(x_2)-\vartheta\left(x_2-\frac{1}{n}\right)\right]}
\]
and taking into account $\vartheta=\vartheta(x)$ is continuous and therefore sequentially continuous, we find that
\[
\lim_{n\to\infty}\|T_n-B^{-1}\|\leq\sqrt{\pi CK\left[\vartheta(0)+\vartheta(x_2)-\vartheta(x_2)\right]}=0
\]
This completes the proof.~~$\square$
\end{proof}
\begin{theorem}
The angular operator $A$ has compact resolvent.
\end{theorem}
\begin{proof}
We know that both $A$ and $\mathcal{B}$ are self-adjoint and therefore their spectra are real. Let us take any $\mu\in\rho(A)$ and $\nu\in\rho(B)$, then the second resolvent equation 
\[
(A-\mu)^{-1}-(\mathcal{B}-\nu)^{-1}=(A-\mu)^{-1}(\mathcal{B}-A+\mu-\nu)(\mathcal{B}-\nu)^{-1}
\]
but $(\mathcal{B}-\nu)^{-1}$ is compact and $(A-\mu)^{-1}$ and $(\mathcal{B}-A+\mu-\nu)$ are bounded, hence the operator on the r.h.s. of the above expression is compact. This implies that also $(A-\mu)^{-1}$ must be compact.~~$\square$
\end{proof}

We conclude this section by deriving a set of necessary and sufficient conditions for the angular eigenvalue problem to have polynomial solutions on the interval $(0,\pi)$. The same conditions can also be used to compute the corresponding eigenvalues. To this purpose we start with (\ref{Smeno}) and make the substitution $S_{-}(\vartheta)=(\sin{\vartheta})^{-1/2}\widetilde{S}_{-}(\vartheta)$ leading to an ODE of the same form as (\ref{Smeno}) with $S_{-}$ and $\widetilde{\mathcal{L}}_{\pm}$ replaced by $\widetilde{S}_{-}$ and 
\[
\widehat{\mathcal{L}}_{\pm}=\frac{d}{d\vartheta}\pm\frac{\Xi}{\Delta_\vartheta}\left(a\omega\sin{\vartheta}+\frac{\widehat{k}}{\sin{\vartheta}}\right),
\] 
respectively. In the case $a=\Lambda=0$ the equation can be solved and the solutions can be expressed in terms of Jacobi polynomials \cite{Davide2,Wink}. If $\Lambda=0$ the same equation can be transformed into a generalized Heun equation \cite{Davide2}. Note that the initial choice of working with Carter's tetrad leads to expressions for the operators $\widehat{\mathcal{L}}_{\pm}$ that are much simpler than the corresponding ones for the operators $\mathcal{L}_{1/2}$ and $\mathcal{L}_{1/2}^\dagger$  appearing in \cite{Davide3} where Kinnersley's tetrad was adopted. Finally, this ODE becomes
\[
\left[\Delta_\vartheta\frac{d^2}{d\vartheta^2}+
\left(\frac{\Delta^{'}_\vartheta}{2}+\frac{am_e\Delta_\vartheta\sin{\vartheta}}{\lambda+am_e\cos{\vartheta}}
\right)\frac{d}{d\vartheta}
-\Xi H^{'}+\frac{\Xi H \Delta^{'}_\vartheta}{2\Delta_\vartheta}-\frac{\Xi^2 H^2}{\Delta_\vartheta}-
\frac{am_e\Xi H\sin{\vartheta}}{\lambda+am_e\cos{\vartheta}}+\lambda^2-a^2 m_e\cos^2{\vartheta}
\right]\widetilde{S}_{-}=0,
\]
where a prime denotes differentiation with respect to $\vartheta$ and 
\[
H(\vartheta)=\frac{\Xi}{\sqrt{\Delta_\vartheta}}\left(a\omega\sin{\vartheta}+\frac{\widehat{k}}{\sin{\vartheta}}\right).
\]
Introducing the variable transformation $x=(1+\cos{\vartheta})/2$ which maps the interval $(0,\pi)$ to the interval $(0,1)$ the above differential equation becomes
\[
\frac{d^2\widetilde{S}_{-}}{dx^2}+\left[-\frac{x-1/2}{x(1-x)}+\frac{\dot{\Delta}_x}{2\Delta_x}-\frac{2am_e}{2am_e x+\lambda-am_e}\right]
\frac{d\widetilde{S}_{-}}{dx}+\left[\frac{\Xi}{4x^2(1-x)^2}\left(\frac{2\widehat{k}x-\widehat{k}}{\Delta_x}-\frac{\Xi\widehat{k}^2}{\Delta_x^2}\right)\right.
\]
\[
\left.+\frac{a\omega\Xi(1-2x)}{x(1-x)\Delta_x}-a\omega\frac{\dot{\Delta}_x}{\Delta_x^2}-\frac{4a^2\omega^2\Xi^2}{\Delta_x^2}-\frac{2am_e\Xi}{2am_e x+\lambda-am_e}\left(\frac{2a\omega}{\Delta_x}+\frac{\widehat{k}}{2x(1-x)\Delta_x}\right)+
\frac{\lambda^2-a^2 m_e^2(2x-1)^2}{x(1-x)\Delta_x}
\right]\widetilde{S}_{-}=0,
\]
where a dot denotes differentiation with respect to $x$ and $\Delta_x=4(\Xi-1)x(x-1)+\Xi$. Let $Q$ denote the variable coefficient multiplying $\widetilde{S}_{-}$ in the above equation. In order to kill terms like $1/x^2$ and $1/(1-x)^2$ appearing in the expression for $Q$ we observe that
\[
\frac{\Xi}{4x^2(1-x)^2}\left(\frac{2\widehat{k}x-\widehat{k}}{\Delta_x}-\frac{\Xi\widehat{k}^2}{\Delta_x^2}\right)=\frac{A_1}{x}+\frac{A_2}{x^2}+\frac{B_1}{1-x}+\frac{B_2}{(1-x)^2}+\frac{\alpha_0+\alpha_1 x}{\Delta_x}+\frac{C_1}{\Delta_x^2}
\]
with
\[
A_1=-\frac{\widehat{k}}{2\Xi}\left[\widehat{k}(5\Xi-4)+2(\Xi-1)\right],\quad
A_2=-\frac{\widehat{k}(\widehat{k}+1)}{4},\quad
B_1=-\frac{\widehat{k}}{2\Xi}\left[\widehat{k}(5\Xi-4)-2(\Xi-1)\right],\quad
B_2=-\frac{\widehat{k}(\widehat{k}-1)}{4},
\]
\[
\alpha_0=-\frac{4\widehat{k}}{\Xi}(2\widehat{k}+1)(\Xi-1)^2,\quad
\alpha_1=\frac{8\widehat{k}(\Xi-1)^2}{\Xi},\quad
C_1=-4\widehat{k}^2(\Xi-1)^2
\]
and we rewrite $Q$ as follows
\[
Q(x)=\frac{A_2}{x^2}+\frac{B_2}{(1-x)^2}+\frac{\Delta_x[A_1+(B_1-A_1)x]+x(1-x)(\alpha_0+\alpha_1 x)+a\omega\Xi(1-2x)}{x(1-x)\Delta_x}+
\]
\[
\frac{C_1-4a^2\omega^2\Xi^2-a\omega\Xi\dot{\Delta}_x}{\Delta_x^2}-\frac{2am_e\Xi}{2am_e x+\lambda-am_e}\left(\frac{2a\omega}{\Delta_x}+\frac{\widehat{k}}{2x(1-x)\Delta_x}\right)+\frac{\lambda^2-a^2 m_e^2(2x-1)^2}{x(1-x)\Delta_x}.
\]
At this point we introduce the s-homotopic transformation $\widetilde{S}_{-}(x)=x^\alpha(1-x)^\beta\Phi(x)$ with $\alpha,\beta\in\mathbb{C}$
and we find that the equation satisfied by $\Phi$ is
\[
\frac{d^2 \Phi}{dx^2}+p(x)\frac{d\Phi}{dx}+q(x)\Phi=0
\]
with
\begin{eqnarray*}
p(x)&=&\frac{2\alpha+1/2-(1+2\alpha+2\beta)x}{x(1-x)}+\frac{\dot{\Delta}_x}{2\Delta_x}-\frac{2am_e}{2am_e x+\lambda-am_e},\\
q(x)&=&\frac{\alpha^2-\alpha/2}{x^2}+\frac{\beta^2-\beta/2}{(1-x)^2}-\frac{2\alpha\beta+\alpha/2+\beta/2}{x(1-x)}
+\frac{\alpha\dot{\Delta}_x}{2x\Delta_x}-\frac{\beta\dot{\Delta}_x}{2(1-x)\Delta_x}+
\frac{2a m_e[(\alpha+\beta)x-\alpha]}{x(1-x)(2am_e x+\lambda-am_e)}+Q(x).
\end{eqnarray*}
The coefficients of $x^{-2}$ and $(1-x)^{-2}$ will vanish whenever
\[
\alpha^2-\frac{\alpha}{2}-\frac{\widehat{k}(\widehat{k}+1)}{4}=0,\quad
\beta^2-\frac{\beta}{2}-\frac{\widehat{k}(\widehat{k}-1)}{4}=0.
\]
The only acceptable roots are those keeping the solution regular at $x=0$ and $x=1$ and they are given as follows
\[
\alpha=\frac{1}{4}+\frac{1}{2}\left|\widehat{k}+\frac{1}{2}\right|=\frac{1}{4}+\frac{|k+1|}{2},\quad
\beta=\frac{1}{4}+\frac{1}{2}\left|\widehat{k}-\frac{1}{2}\right|=\frac{1}{4}+\frac{|k|}{2}.
\]
Finally, multiplying the equation for $\Phi$ by $x(1-x)\Delta_x^2(2am_e x+\lambda-am_e)$ yields
\begin{equation}\label{zwei}
X(x)\frac{d^2\Phi}{dx^2}+Y(x)\frac{d\Phi}{dx}+Z(x)\Phi=0,
\end{equation}
where $X$, $Y$, and $Z$ are polynomials of degree $7$, $6$, and $5$, respectively, given by
\[
X(x)=x(1-x)\Delta_x^2(2am_e x+\lambda-am_e)=\sum_{j=0}^7 a_j x^j
\]
with
\[
a_7=-32 am_e(\Xi-1)^2,\quad a_6=-16(\Xi-1)^2(\lambda-7am_e),\quad a_5=-16 am_e\Xi(\Xi-1)+48(\Xi-1)^2(\lambda-3am_e),
\]
\[
a_4=-8\Xi(\Xi-1)(\lambda-5am_e)-16(\Xi-1)^2(3\lambda-5am_e),\quad
a_3=-2am_e\Xi^2+16\Xi(\Xi-1)(\lambda-2am_e)+16(\Xi-1)^2(\lambda-am_e),
\]
\[
a_2=-\Xi^2(\lambda-3am_e)-8\Xi(\Xi-1)(\lambda-am_e),\quad
a_1=\Xi^2(\lambda-am_e),\quad a_0=0,
\]
\[
Y(x)=X(x)p(x)=\sum_{j=0}^6 b_j x^j
\]
with
\[
b_6=-32am_e(\Xi-1)^2(1+2\alpha+2\beta),\quad
b_5=32(\Xi-1)^2[am_e(7\alpha+5\beta+3)-\lambda(1+\alpha+\beta)],
\]
\[
b_4=-8(\Xi-1)\left\{(\Xi-1)[5am_e(8\alpha+4\beta+3)-2\lambda(6\alpha+4\beta+5)]+am_e(4\alpha+4\beta+1)\right\},
\]
\[
b_3=4(\Xi-1)\left\{(\Xi-1)[20am_e(3\alpha+\beta+1)-\lambda(28\alpha+12\beta+19)]+4am_e(5\alpha+3\beta+1)-\lambda(4\alpha+4\beta+3)\right\},
\]
\[
b_2=-4am_e(\alpha+\beta)+(\Xi-1)\left\{\Xi[-10am_e(10\alpha+2\beta+3)+2\lambda(32\alpha+8\beta+17)]+4am_e(7\alpha-\beta+4)-16\lambda(1+2\alpha)
\right\},
\]
\[
b_1=2am_e(3\alpha+\beta)-\lambda(2\alpha+2\beta+1)+(\Xi-1)\left\{
\Xi[2am_e(11\alpha+\beta+3)-\lambda(18\alpha+2\beta+7)]+2am_e(3\alpha+\beta)-\lambda(2\alpha+2\beta+1)
\right\},
\]
\[
b_0=\frac{\Xi^2}{2}(4\alpha+1)(\lambda-am_e),
\]
and
\[
Z(x)=X(x)q(x)=\sum_{j=0}^5 c_j x^j
\]
with
\[
c_5=-16am_e(\Xi-1)[(\Xi-1)(\widehat{k}^2+\alpha+\beta+4\alpha\beta)+2a^2m_e^2],
\]
\[
c_4=8(\Xi-1)\left\{
(\Xi-1)[\widehat{k}^2(5am_e-\lambda)+am_e(2\widehat{k}+2a\omega+5\alpha+5\beta+20\alpha\beta)\right.
\]
\[
\left.-\lambda(4\alpha\beta+3\alpha+3\beta)]+
2a^2m_e(5am_e^2-\lambda m_e+\omega)
\right\},
\]
\[
c_3=8a^3m_e(\omega^2-m_e^2)+(\Xi-1)\left\{\Xi[8\lambda(2\widehat{k}^2+\widehat{k}+8\alpha\beta+7\alpha+5\beta)-8am_e(6\widehat{k}^2+4\widehat{k}+20\alpha\beta+5\alpha+5\beta)+\right.
\]
\[
\left.8a^2 m_e\omega(a\omega-4)]-8\lambda(2\widehat{k}^2+\widehat{k}+8\alpha\beta+7\alpha+5\beta)+8am_e(5\widehat{k}^2+4\widehat{k}+16\alpha\beta+5\alpha+5\beta+\lambda^2)+\right.
\]
\[
\left.8a^2 m_e(4\lambda m_e+a\omega^2-10am_e^2)\right\},
\]
\[
c_2=4a^2(\lambda-3am_e)(\omega^2-m_e^2)+(\Xi-1)\left\{\Xi[-4\lambda(4\widehat{k}^2+3\widehat{k}+12\alpha\beta+12\alpha+6\beta)+4am_e(8\widehat{k}^2+5\widehat{k}+20\alpha\beta+5\alpha+5\beta)+\right.
\]
\[
\left.4a^2\omega(4m_e+\lambda\omega-3am_e\omega)]+4\lambda(3\widehat{k}^2+3\widehat{k}+\lambda^2+8\alpha\beta+10\alpha+4\beta)-4am_e(5\widehat{k}^2+5\widehat{k}+8\alpha\beta+6\alpha+4\beta+3\lambda^2)+\right.
\]
\[
\left. 4a^2[\lambda(\omega^2-6m_e^2)-am_e(3\omega^2-10m_e^2)]\right\},
\]
\[
c_1=-4a^2\lambda(\omega^2-m_e^2)-2a^3 m_e(3m_e^2-2\omega^2)-am_e(\widehat{k}^2+4\alpha\beta-\alpha-\beta-2\lambda^2)-2\lambda a\omega+
\]
\[
(\Xi-1)\left\{\Xi[2\lambda(4\widehat{k}^2+3\widehat{k}+8\alpha\beta+9\alpha+3\beta)-am_e(13\widehat{k}^2+4\widehat{k}+20\alpha\beta+5\alpha+5\beta)+2a\omega(2a^2\omega m_e-2a\omega\lambda-\lambda)]\right.
\]
\[
\left.-4\lambda(\widehat{k}^2+\widehat{k}+\lambda^2+2\alpha)+am_e(3\widehat{k}^2+4\widehat{k}-4\alpha\beta+9\alpha+\beta+6\lambda^2)
 -2a(\omega\lambda+5a^2 m_e^3-2a^2\omega^2 m_e-4a\lambda m_e^2+2a\lambda\omega^2)\right\},
\]
\[
c_0=\frac{\Xi}{2}\left\{
(\Xi-1)[-\lambda(5\widehat{k}^2+2\widehat{k}+4\alpha\beta+5\alpha+\beta)+am_e(5\widehat{k}^2+4\alpha\beta+\alpha+\beta)+2a\omega(\lambda-am_e)]
\right.
\]
\[
\left.-\lambda(\widehat{k}^2-2\lambda^2+4\alpha\beta+\alpha+\beta)+am_e(\widehat{k}^2-2\widehat{k}+4\alpha\beta-3\alpha+\beta-2\lambda^2)+2a(\lambda-am_e)(\omega-am_e^2)\right\}.
\]
Note that the coefficients $a_7,\cdots,a_0$ are not algebraically independent since their sum vanishes. 
\begin{theorem}
If $a=0$, equation (\ref{zwei}) admits a constant polynomial solution.
\end{theorem}
\begin{proof}
Without loss of generality let $\Phi(x)=1$. This will be a solution of (\ref{zwei}) whenever $Z(x)=0$, that is $c_5=\cdots=c_0=0$. From the expression giving $c_5$ we see that it will vanish if $a=0$ for which $\Xi=1$ or $m_e=0$. In the case $a=0$ all coefficients $c_i$ vanish identically and therefore $\Phi(x)=1$ is a solution of (\ref{zwei}). If $m_e=0$ and $\widehat{k}\geq 1/2$, we find that $c_5=-4(\Xi-1)^2(2\widehat{k}+3)(2\widehat{k}+1)$ can never vanish. The same situation occurs for $\widehat{k}\leq-1/2$.~~$\square$
\end{proof}
\begin{remark}
The case $a=0$ corresponds to the Schwarzschild-deSitter metric for which one can show that (\ref{zwei}) becomes a hypergeometric equation and the eigenvalues can be computed explicitely as given in \cite{Davide3,Wink}.
\end{remark}
To find a set of necessary and sufficient conditions for the existence of polynomial solutions of the form $\Phi(x)=x-x_1$, we generalize the so-called functional (or analytic) Bethe Ansatz method used in \cite{Zhang} so that it can be applied to equation (\ref{zwei}). 
\begin{theorem}
The polynomial $\Phi(x)=x-x_1$ is a solution of (\ref{zwei}) if and only if 
\begin{eqnarray*}
&&b_6+c_5=0,\\
&&b_6 x_1+b_5+c_4=0,\\
&&b_6 x_1^2+b_5 x_1+b_4+c_3=0,\\
&&b_6 x_1^3+b_5 x_1^2+b_4 x_1+b_3+c_2=0,\\
&&b_6 x_1^4 +b_5 x_1^3+b_4 x_1^2+b_3 x_1+b_2+c_1=0,
\end{eqnarray*}
where $x_1$ is determined by the equation $Y(x_1)=0$.
\end{theorem}
\begin{proof}
First of all, by replacing the ansatz $\Phi(x)=x-x_1$ into (\ref{zwei}) we find that $Y(x)+Z(x)(x-x_1)=0$ and hence
\[
-c_0=(x-x_1)^{-1}\sum_{i=0}^6 b_i x^i+\sum_{i=0}^5 c_i x^i.
\]
The l.h.s. is a constant while the r.h.s. is a meromorphic function with simple pole at $x=x_i$ and a singularity at $x=\infty$. The residue of $-c_0$ at $x=x_1$ is
\[
\mbox{Res}(-c_0,x=x_1)=\lim_{x\to x_1}(x-x_1)(-c_0)=\sum_{i=0}^6 b_i x^i_1=Y(x_1).
\]
Then, we have
\[
-c_0-(x-x_1)^{-1}\mbox{Res}(-c_0,x=x_1)=(x-x_1)^{-1}\sum_{i=0}^6 b_i(x^i-x_1^i)+\sum_{i=0}^5 c_i x^i
\]
and after simplification of $x-x_1$ in the denominator of the term appearing on the r.h.s. of the above expression we end up with
\begin{equation}\label{zvezda}
-c_0=\sum_{i=1}^5\gamma_i x^i+\frac{Y(x_1)}{x-x_1}+K
\end{equation}
with
\[
\gamma_1=\sum_{i=0}^4 b_{i+2}x^i_1+c_1,\quad
\gamma_2=\sum_{i=0}^3 b_{i+3}x^i_1+c_2,\quad
\gamma_3=\sum_{i=0}^2 b_{i+4}x^i_1+c_3,
\]
\[
\gamma_4=b_6 x_1+b_5+c_4,\quad
\gamma_5=b_6+c_5,\quad
K=\sum_{i=0}^4 b_{i+1}x^i_1.
\]
The r.h.s. of (\ref{zvezda}) is a constant if and only if the $\gamma_i$'s as well as $Y(x_1)$ vanish. Note that when this set of equations are satisfied, then $-c_0=K$.~~$\square$
\end{proof}
The same strategy can be applied to show that $\Phi(x)=(x-x_1)(x-x_2)$ will be a solution of (\ref{zwei}) whenever the following set of conditions are satisfied, namely
\[
2(a_7+b_6)+c_5=0,\quad (2a_7+b_6)(x_1+x_2)+2(a_6+b_5)+c_4=0,
\]
\[
2(a_7+a_5)(x_1^2+x_1 x_2+x_2^2)+(2a_6+b_5)(x_1+x_2)+b_6(x_1^2+x_2^2)+2(a_5+b_4)+c_3=0,
\]
\[
(x_1+x_2)[2a_7(x_1^2+x_2^2)+2a_5+b_4]+2a_6(x_1^2+x_1 x_2+x_2^2)+b_6(x_1^3+x_2^3)+b_5(x_1^2+x_2^2)+2(a_4+b_3)+c_2=0,
\]
\[
2a_7(x_1^4+x_1^3 x_2+x_1^2 x_2^2+x_1 x_2^3+x_2^4)+(x_1+x_2)[2a_6(x_1^2+x_2^2)+2a_4+b_3]+b_6(x_1^4+x_2^4)+
\]
\[
b_5(x_1^3+x_2^3)+b_4(x_1^2+x_2^2)+2(a_3+b_2)+c_1=0,
\]
\[
2a_7(x_1^5+x_1^4 x_2+x_1^3 x_2^2+x_1^2 x_2^3+x_1 x_2^4+x_2^5)+2a_6(x_1^4+x_1^3 x_2+x_1^2 x_2^2+x_1 x_2^3+x_2^4)+(x_1+x_2)[2a_5(x_1^2+x_2^2)+2a_3+b_2]+
\]
\[
2a_4(x_1^2+x_1 x_2+x_2^2)+b_6(x_1^5+x_2^5)+b_5(x_1^4+x_2^4)+b_4(x_1^3+x_2^3)+b_3(x_1^2+x_2^2)+2(a_2+b_1)+c_0=0
\]
together with $\mbox{Res}(-c_0,x=x_n)=0$ for $n=1,2$. This equation gives rise to two algebraic equations for the roots $x_n$, more precisely
\[
\sum_{m=1 \atop m\neq n}^2\frac{2}{x_m-x_n}+\frac{Y(x_n)}{X(x_n)}=0.
\]

\section{The radial system}
In this section we show that the Dirac equation in the non-extreme Kerr-deSitter metric does not allow for bound state solutions by computing the deficiency index of the radial operator associated to the system (\ref{radial}). More precisely, since the deficiency index counts the number of square integrable solutions, it suffices to show that the deficiency index of the radial operator vanishes. We start by observing that the components of the radial spinor satisfy the relations $R_{-}=\overline{R}_{+}$ and $R_{+}=\overline{R}_{-}$ as it can be easily verified from (\ref{radial}). This property motivates the ansatz $R_{-}(r)=F(r)-iG(r)$ and $R_{+}(r)=F(r)+iG(r)$ used in the proof of the following theorem.
\begin{theorem}
In the non-extreme Kerr-deSitter metric the Dirac equation does not possess bound state solutions.
\end{theorem}
\begin{proof}
The proof relies on an application of Theorem~$5.2$ in \cite{Lesch} and follows the strategy adopted in the proof of Theorem~$5.1$ in \cite{BN}. To this purpose we cast the radial system (\ref{radial}) into the form
\begin{equation}\label{rs}
\left(
\begin{array}{cc}
\sqrt{\Delta_r}\frac{d}{dr}+i\frac{V(r)}{\sqrt{\Delta_r}} & -im_e r-\lambda\\
im_e r-\lambda & \sqrt{\Delta_r}\frac{d}{dr}+i\frac{V(r)}{\sqrt{\Delta_r}}
\end{array}
\right)\left(\begin{array}{c}
R_{-}\\
R_{+}
\end{array}\right)=0,\quad V(r)=\Xi\left[\omega(r^2+a^2)+a\widehat{k}\right].
\end{equation}
In order to bring (\ref{rs}) into a Dirac system we transform the dependent variable according to $R_{-}(r)=F(r)-iG(r)$ and $R_{+}(r)=F(r)+iG(r)$ and we introduce a new independent variable $u$ defined through the relation
\begin{equation}\label{ve}
\frac{du}{dr}=\frac{\Xi(r^2+a^2)}{\Delta_r}.
\end{equation}
Then, (\ref{rs}) becomes
\begin{equation}\label{dir_sys}
(\mathfrak{U}\Omega)(u):=J\frac{d\Omega}{du}+B(u)\Omega=\omega\Omega
\end{equation}
with $\Omega=(F,G)^T$ and
\[
J=\left(\begin{array}{cc}
0 & 1\\
-1 & 0
\end{array}
\right),\quad
B(u)=\left(\begin{array}{cc}
\frac{m_e r\sqrt{\Delta_r}-a\Xi\widehat{k}}{\Xi(r^2+a^2)} & \frac{\lambda\sqrt{\Delta_r}}{\Xi(r^2+a^2)}\\
\frac{\lambda\sqrt{\Delta_r}}{\Xi(r^2+a^2)} & -\frac{m_e r\sqrt{\Delta_r}+a\Xi\widehat{k}}{\Xi(r^2+a^2)}
\end{array}
\right),
\]
where the radial variable $r$ is now a function of the tortoise coordinate $u$. In the non-extreme case the behaviour of the solution of (\ref{ve}) in proximity of the event and cosmological horizons is captured by the following formulae
\begin{equation}\label{asymp1}
u(r)=\frac{1}{2\kappa_{+}}\ln{(r-r_{+})}+\mathcal{O}(r-r_{+}), \quad
u(r)=\frac{1}{2\kappa_{c}}\ln{(r_c-r)}+\mathcal{O}(r-r_c),
\end{equation}
where $\kappa_{+}$ and $\kappa_{c}$ denote the surface gravity at the event and cosmological horizon, respectively, and are given by
\[
\kappa_{+}=\frac{3(r_{+}-M)-\Lambda r_{+}(2r_{+}^2+a^2)}{3\Xi(r_{+}^2+a^2)},\quad
\kappa_{c}=\frac{\Lambda r_{c}(2r_{c}^2+a^2)-3(r_{c}-M)}{3\Xi(r_{c}^2+a^2)}.
\]
Since $\Delta_r>0$ on the interval $(r_{+},r_c)$, it follows from (\ref{ve}) that $u$ must be an increasing function of $r$ and therefore, the numerator in the expression of $\kappa_{+}$ is positive. This means that $u\to-\infty$ as $r\to r_{+}$. By the same token we can also conclude that the numerator in the expression of $\kappa_c$ must be negative and therefore, $u\to+\infty$ as $r\to r_c$. Furthermore, $u$ maps the interval $(r_{+},r_c)$  to the real line. Taking into account that the radial spinor $R=(R_{-},R_{+})^T$ is square integrable on $(r_{+},r_c)$ if
\[
\langle R,R\rangle=\Xi\int_{r_{+}}^{r_c} dr~\frac{r^2+a^2}{\Delta_r}|R|^2=2\Xi\int_{r_{+}}^{r_c} dr~\frac{r^2+a^2}{\Delta_r}R_{-}R_{+}<+\infty,
\]
it results that the transformed spinor $\Omega$ is square integrable whenever
\begin{equation}\label{ip}
\langle\Omega,\Omega\rangle=\int_{-\infty}^{+\infty} du~\left[F(u)^2+G(u)^2\right]<+\infty.
\end{equation}
The formal differential operator $\mathfrak{U}$ is formally symmetric since $J=-J^{*}$ and $B=B^{*}$ where star denotes complex transposition. Let $S_{min}$ be the minimal operator associated to $\mathfrak{U}$ such that $S_{min}$ acts on the Hilbert space $L^2(\mathbb{R},du)^2$ equipped with the inner-product (\ref{ip}). Then, $S_{min}$ with domain of definition $D(S_{min})=C^\infty_0(\mathbb{R})^2$ such that $S_{min}\Omega:=\mathfrak{U}\Omega$ for $\Omega\in D(S_{min})$ is densely defined and closable. We denote by $S$ the closure of $S_{min}$ and apply Neumark's decomposition method \cite{Neumark}. To this purpose, let $S_{min,\pm}$ be the minimal operators associated to $\mathfrak{U}$ when restricted on the half-lines $[0,+\infty)$ and $(-\infty,0]$, respectively. We consider $S_{min,\pm}$ on $L^2(\mathbb{R}_\pm,du)^2$ equipped with (\ref{ip}). The operators $S_{min,\pm}$ with domain of definition $D(S_{min,\pm})=C_0^\infty(\mathbb{R}_\pm)^2$ and $S_{min,\pm}\Omega_\pm:=\mathfrak{U}\Omega_\pm$ for $\Omega_\pm\in D(S_{min,\pm})$ are densely defined and closable. Furthermore, $\mathfrak{U}$ is in the limit point case at $\pm\infty$. This can be seen as follows. First of all, we recall that since the limit point and limit circle cases are mutually exclusive, we can determine the appropriate case if we examine the solution of (\ref{dir_sys}) for a single value of $\omega$. Hence, without loss of generality we set $\omega=0$ and consider the system
\begin{equation}\label{dirac_sys_2}
J\frac{d\Omega}{du}+B(u)\Omega=0.
\end{equation}
First of all, we observe that the matrix $B(u)$ converges for $u\to-\infty$ to the constant matrix
\[
B_0=\lim_{u\to-\infty}B(u)=\left(\begin{array}{cc}
-\frac{a\widehat{k}}{r_{+}^2+a^2} & 0\\
0 & -\frac{a\widehat{k}}{r_{+}^2+a^2}
\end{array}
\right).
\]
This observation together with (\ref{asymp1}) suggests an asymptotic expansion of $B(u)$ in powers of $e^{\kappa_{+}u}$. A straightforward computation gives
\[
B(u)=B_0+\widetilde{B}(u),\quad |\widetilde{B}(u)|\leq C e^{\kappa_{+}u}
\]
for some positive constant $C$. Finally, applying Theorem~$1$ in Ch. $IV$ in \cite{Coppel} yields that the system (\ref{dirac_sys_2}) has asymptotic solutions for $u\to-\infty$ given by
\begin{equation}\label{IaIIa}
F(u)=\cos{\left(\frac{a\widehat{k}}{r_{+}^2+a^2}u\right)}\left(\begin{array}{c}
1\\
0
\end{array}\right)+\mathcal{O}\left(e^{\kappa_{+}u}\right),\quad
G(u)=\sin{\left(\frac{a\widehat{k}}{r_{+}^2+a^2}u\right)}\left(\begin{array}{c}
0\\
1
\end{array}\right)+\mathcal{O}\left(e^{\kappa_{+}u}\right).
\end{equation}
The problem at the cosmological horizon can be treated similarly and for $u\to+\infty$ we find that
\begin{equation}\label{IIIaIVa}
F(u)=\cos{\left(\frac{a\widehat{k}}{r_{c}^2+a^2}u\right)}\left(\begin{array}{c}
1\\
0
\end{array}\right)+\mathcal{O}\left(e^{\kappa_{c}u}\right),\quad
G(u)=\sin{\left(\frac{a\widehat{k}}{r_{c}^2+a^2}u\right)}\left(\begin{array}{c}
0\\
1
\end{array}\right)+\mathcal{O}\left(e^{\kappa_{c}u}\right).
\end{equation}
By inspecting (\ref{IaIIa}) and (\ref{IIIaIVa}) we can immediately conclude that the differential operator $\mathfrak{U}$ is in the limit point case at $\pm\infty$. This implies that the operators $S_{min,\pm}$ are essentially self-adjoint. Let $S_{\pm}$ denote the closure of $S_{min,\pm}$ and $N_{\pm}(S_\pm)$ the corresponding deficiency indices. If $\nu_\pm$ denotes the number of positive and negative eigenvalues of the matrix $iJ$, then, since $\nu_+=1=\nu_-$, Theorem~$5.2$ in \cite{Lesch} implies that $N_{\pm}(S_+)=1=N_{\pm}(S_-)$. Since zero is the only solution of (\ref{dirac_sys_2}) in $L^2(\mathbb{R},du)$ the orginal system (\ref{dir_sys}) is definite on $\mathbb{R}_+$ and $\mathbb{R}_{-}$ in the sense of \cite{Lesch}. Finally, $(5.11a)$ in Proposition~$5.4$ in \cite{Lesch} yields that the deficiency indices for $S$ are 
\[
N_\pm(S)=N_\pm(S_+)+N_\pm(S_-)-2=0.
\]
This implies that the radial system (\ref{dir_sys}) does not possess any square integrable solution on the whole real line and this completes the proof.~~$\square$
\end{proof}

\appendix
\section{Derivation of formulae (\ref{r12}) and (\ref{r34})}
First of all, the quartic equation (\ref{**}) is already in reduced form. Let $r_i$ with $i=1,\cdots,4$ denote the roots of (\ref{**}).
 Applying Vieta's formulae we find the following useful relations among the roots of $\Delta_r$, namely
\[
r_{1}+r_{2}+r_{3}+r_4=0,\quad
r_{1}(r_{2}+r_{3}+r_4)+r_{2}(r_{3}+r_4)+r_3 r_4=p,
\]
\[
r_{1}r_{2}r_{3}r_4=u,\quad
r_{1}(r_{2}r_{3}+r_{2}r_4+r_{3}r_{4})+r_{2}r_{3}r_4=-q
\]
with $p$, $q$, $u$ defined in (\ref{pqu}). Let us consider the following polynomials in $r_i$, namely
\[
z_1=(r_1+r_2)(r_3+r_4),\quad z_2=(r_1+r_3)(r_2+r_4),\quad z_3=(r_1+r_4)(r_2+r_3).
\]
Using Vieta's formulae we find that
\begin{eqnarray}
z_1+z_2+z_3&=&2(r_1 r_2+r_1 r_3+r_1 r_4+r_2 r_3+r_2 r_4+r_3 r_4)=2p,\label{z1}\\
z_1z_2+z_2z_3+z_3z_1&=&p^2-4u,\quad z_1z_2z_3=-q^2 \label{z2}.
\end{eqnarray}
This means that $z_1$, $z_2$, and $z_3$ are roots of the so-called resolvent cubic $z^3-2pz^2+(p^2-4u)z+q^2=0$. Since (\ref{**}) is 
reduced, Vieta's formula $r_1+r_2+r_3+r_4=0$ holds and equations (\ref{z1}) and (\ref{z2}) can be solved yielding $r_1+r_2=\sqrt{-z_1}$, $r_3+r_4=-\sqrt{-z_1}$, $r_1+r_3=\sqrt{-z_2}$, $r_2+r_4=-\sqrt{-z_2}$, $r_1+r_4=\sqrt{-z_3}$, and $r_2+r_3=-\sqrt{-z_3}$. The ambiguity in the choice of the sign of the square root can be fixed according to the relation $\sqrt{-z_1}\sqrt{-z_2}\sqrt{-z_3}=-q$. Taking into account that in the present case $q>0$ we choose the sign of the square root so that
\begin{eqnarray*}
2r_1&=&-\sqrt{-z_1}-\sqrt{-z_2}-\sqrt{-z_3},\quad 2r_2=-\sqrt{-z_1}+\sqrt{-z_2}+\sqrt{-z_3},\\
2r_3&=&+\sqrt{-z_1}-\sqrt{-z_2}+\sqrt{-z_3},\quad 2r_4=+\sqrt{-z_1}+\sqrt{-z_2}-\sqrt{-z_3}.
\end{eqnarray*}
To find the roots $z_1$, $z_2$, and $z_3$ of the resolvent cubic we reduce it to the special cubic $w^3+\widehat{p} w+\widehat{q}=0$ by means of the Tschirnhaus transformation $z=w+(2p)/3$ where
\[
\widehat{p}=-\left(\frac{p^2}{3}+4u\right),\quad \widehat{q}=\frac{2}{27}p^3-\frac{8}{3}pu+q^2.
\]
Let $w_1$, $w_2$, and $w_3$ denote the roots of the special cubic. With the help of Vieta's formulae  
\[
w_1+w_2+w_3=0,\quad w_1 w_2+w_2 w_3+w_3 w_1=\widehat{p},\quad w_1 w_2 w_3=-\widehat{q}
\]
the discriminant of the special cubic turns to be
\[
D=\prod_{1\leq j<\ell\leq 3}(w_j-w_\ell)^2=-4(w_1w_2+w_2w_3+w_3w_1)^3-27(w_1w_2w_3)^2=-4\widehat{p}^3-27\widehat{q}^2.
\]
Let $\omega=e^{2\pi i/3}$ be the primitive root of unity and introduce the Lagrange substitutions
\[
\xi_k=\frac{1}{3}\sum_{j=1}^3 \omega^{(j-1)k}w_j,\quad k=0,1,2.
\]
Using $\omega^3=1$ and $w_1+w_2+w_3=0$ we find that
\[
3\xi_0=w_1+w_2+w_3=0,\quad
3\xi_1=w_1+\omega w_2+\omega^2 w_3,\quad
3\xi_2=w_1+\omega^2 w_2+\omega w_3.
\]
A somewhat tedious computation gives
\[
\xi_1^3=-\frac{\widehat{q}}{2}+i\frac{\sqrt{3}}{18}\sqrt{D}=-\frac{\widehat{q}}{2}+\sqrt{\left(\frac{\widehat{p}}{3}\right)^3-\left(\frac{\widehat{q}}{2}\right)^2},\quad
\xi_2^3=-\frac{\widehat{q}}{2}-i\frac{\sqrt{3}}{18}\sqrt{D}=-\frac{\widehat{q}}{2}-\sqrt{\left(\frac{\widehat{p}}{3}\right)^3-\left(\frac{\widehat{q}}{2}\right)^2},
\]
where we used the fact that $\omega^2+\omega+1=0$. Hence, we obtain 
\[
\xi_1=\sqrt[3]{-\frac{\widehat{q}}{2}+\sqrt{\left(\frac{\widehat{p}}{3}\right)^3-\left(\frac{\widehat{q}}{2}\right)^2}},\quad
\xi_2=\sqrt[3]{-\frac{\widehat{q}}{2}-\sqrt{\left(\frac{\widehat{p}}{3}\right)^3-\left(\frac{\widehat{q}}{2}\right)^2}}
\]
where the signs of the cubic roots must be chosen so that $\xi_1\xi_2=-\widehat{p}/3$. The Lagrange substitutions introduced earlier can be solved uniquely in terms of $\xi_1$ and $\xi_2$ and we obtain
\[
w_j=\sum_{k=1}^2\omega^{-(j-1)k}\xi_k,\quad j=1,2,3.
\] 
Finally, we get the roots of the special cubic as
\[
w_1=\xi_1+\xi_2,\quad
w_2=\omega^2\xi_1+\omega\xi_2,\quad
w_3=\omega\xi_1+\omega^2\xi_2.
\]
\section{The extreme case}
Let $\mu=\mu_{1,+}$ and factorize the polynomial in (\ref{***}) as $(\rho-\rho_{--})(\rho-\rho_{-})(\rho-\rho_c)^2$ which compared with (\ref{***}) gives rise to the following system of equations for the unknowns $\rho_{--}$, $\rho_{-}$, and $\rho_c$
\[
\rho_{--}+\rho_{-}+2\rho_c=0,\quad
\rho_{--}\rho_{-}+2(\rho_{--}+\rho_{-})\rho_c+\rho_c^2=\alpha^2-1,
\]
\[
-2\rho_{--}\rho_{-}\rho_c-(\rho_{--}+\rho_{-})\rho_c^2=\mu_{1,+},\quad
\rho_{--}\rho_{-}\rho_c^2=-\alpha^2.
\]
Note that since $\rho_{--}<0$ and $\alpha^2\geq 0$, the last equation implies that $\rho_{-}>0$. Moreover, the first equation permits to express $\rho_{--}$ in terms of $\rho_{-}$ and $\rho_c$ so that the above system reduces to
\begin{eqnarray}
\rho_{-}(\rho_{-}+2\rho_c)+3\rho_c^2&=&1-\alpha^2,\label{I}\\
2\rho_{-}\rho_c(\rho_{-}+2\rho_c)+2\rho_c^3&=&\mu_{1,+},\label{II}\\
\rho_{-}\rho_c^2(\rho_{-}+2\rho_c)&=&\alpha^2.\label{III}
\end{eqnarray}
We have to deal with an overdetermined system and special care must be exercised in analyzing its solutions. First of all, we note that all three equations above are quadratic polynomials in $\rho_{-}$. Solving (\ref{I})-(\ref{III}) with respect to $\rho_{-}$ and taking into account that $\rho_{-}>0$, we find 
\[
\rho_{-,I}=-\rho_c+\sqrt{1-\alpha^2-2\rho_c^2},\quad
\rho_{-,II}=\frac{-2\rho_c^2+\sqrt{2\mu_{1,+}\rho_c}}{2\rho_c},\quad
\rho_{-,III}=\frac{-\rho_c^2+\sqrt{\rho_c^4+\alpha^2}}{\rho_c}.
\]
At this point we must require that $\alpha^2+2\rho_c^2<1$.  This condition is always satisfied because $\rho_{-,I}>0$ and therefore 
$\sqrt{1-\alpha^2-2\rho_c^2}>\rho_c$. Squaring we get $1-\alpha^2>3\rho_c^2$ and hence $1-\alpha^2-2\rho_c^2>3\rho_c^2-2\rho_c^2=\rho_c^2>0$. Furthermore, the Roman numeral attached to the roots indicates that $\rho_{-,I}$ is a root of (\ref{I}) and so on.  In order to make the system (\ref{I})-(\ref{III}) consistent, we must require that $\rho_{-,I}=\rho_{-,II}=\rho_{-,III}$. This implies that
\begin{equation}\label{IV}
\sqrt{1-\alpha^2-2\rho_c^2}=\sqrt{\rho_c^2+\frac{\alpha^2}{\rho_c^2}}=\sqrt{\frac{\mu_{1,+}}{2\rho_c}}.
\end{equation}
From the first equality in (\ref{IV}) we obtain the biquadratic equation 
\begin{equation}\label{V}
3\rho_c^4-(1-\alpha^2)\rho_c^2+\alpha^2=0.  
\end{equation}
Since $\rho_c>0$, the only two acceptable solutions are
\[
\rho_{c,\pm}=\frac{1}{6}\sqrt{6-6\alpha^2\pm 6\sqrt{(\alpha^2-\alpha_{-}^2)(\alpha^2-\alpha_{+}^2)}}
\]
with $\alpha_{\pm}=2\pm\sqrt{3}$. Equating the first term and last term in (\ref{IV}) and then equating the second and last term in (\ref{IV}) we find that
\[
\mu_{1,+}=2\left(\rho_c^3+\frac{\alpha^2}{\rho_c}\right),\quad \mu_{1,+}=2(1-\alpha^2)\rho_c-4\rho_c^3. 
\]
We see immediately that the above expressions for $\mu_{1,+}$ are equivalent provided that $\rho_c$ is a root of (\ref{V}). Finally, a lengthy but straightforward computation shows that $\mu_{1,+}$ coincides with (\ref{mu1}) only if we choose $\rho_c=\rho_{c,+}$. The solution $\rho_c=\rho_{c,-}$ must be disregarded because would lead to the contradiction $\mu_{1,+}=\mu_{2,+}$. Concerning the case $\rho_{-}=\rho_{+}$ and $\mu=\mu_{2,+}$ we observe that the corresponding set of equations involving $\rho_{+}$ and $\rho_c$ can be obtained directly from (\ref{I})-(\ref{III}) with $\rho_{-}$, $\rho_c$, and $\mu_{1,+}$ replaced by $\rho_c$, $\rho_{+}$, and $\mu_{2,+}$, respectively. The rest of the proof is too similar to the previous case to be presented here.


\begin{thebibliography}{999}
\bibitem{Mat}
S. Akcay and R. Matzner, ``Kerr-de Sitter universe'', {\it{Class. Quant. Grav.}} {\bf{28}} (2011) 085012
\bibitem{WMAP1}
E. Komatsu et al., ``Seven-Year Wilkinson Microwave Anisotropy Probe (WMAP) Observations: Cosmological Interpretation'', {\it{Astrophys. J. Suppl.}} {\bf{192}} (2011) 18
\bibitem{WMAP2}
J. Dunkley et al., ``Five-Year Wilkinson Microwave Anisotropy Probe (WMAP) Observations: Likelihoods and Parameters from the WMAP Data'', {\it{Astrophys. J. Suppl.}} {\bf{180}} (2009) 306
\bibitem{WMAP3}
D. N. Spergel et al., ``Three-Year Wilkinson Microwave Anisotropy Probe (WMAP) Observations: Implications for Cosmology'', {\it{Astrophys. J. Suppl.}} {\bf{170}} (2007) 377
\bibitem{WMAP4}
C. L. Bennett et al., ``First-Year Wilkinson Microwave Anisotropy Probe (WMAP) Observations: Preliminary Maps and Basic Results'', {\it{Astrophys. J. Suppl.}} {\bf{148}} (2003) 1
\bibitem{Khan}
U. Khanal, ``Rotating black hole in asymptotic de Sitter space: perturbation of the spacetime with spin fields'', {\it{Phys. Rev. D}} 
{\bf{28}} (1983) 1291
\bibitem{Oota}
T. Oota and Y. Yasui, ``Separability of Dirac equation in higher dimensional Kerr-NUT-de Sitter spacetime'', {\it{Phys. Lett. B}}{\bf{659}} (2008) 688
\bibitem{Kamran}
N. Kamran and R.G. McLenaghan, ``Separation of variables and symmetry operators for the neutrino and Dirac equations in the space-times admitting a two-parameter Abelian orthogonally transitive isometry group and a pair of shear-free geodesic null congruences'', {\it{J. Math. Phys.}} {\bf{25}} (1984) 1019
\bibitem{Bel}
F. Belgiorno and S. L. Cacciatori, ``Absence of time-periodic solutions for the Dirac equation in Kerr-Newman-de Sitter black hole background'', {\it{J. Phys. A: Math. Theor.}} {\bf{42}} (2009) 135207
\bibitem{Davide1}
D. Batic and H. Schmid, ``The Dirac Propagator in the Kerr-Newman Metric'', {\it{Progr. Theor. Phys.}} {\bf{116}} (2006) 517
\bibitem{Cart1}
B. Carter, ``Hamilton-Jakobi and Schr\"{o}dinger separable solutions of Einstein equations'', {\it{Comm. Math. Phys.}} {\bf{10}} (1968) 280
\bibitem{Cart} 
B. Carter, {\it{Black holes/Les astres occlus}}, C. de Witt, B.S. de Witt (Eds.), Proc. of the Les Houches Summer School, 1972, Gordon and Breach, New York (1973)
\bibitem{Arnon}
D. S. Arnon, ``Geometric Reasoning with Logic and Algebra'', {\it{Artificial Intelligence}} {\bf{37}} (1988) 37
\bibitem{Yang}
L. Yang, ``Recent Advances on Determining the Number of Real Roots of Parametric Polynomials'', {\it{J. Symb. Comp.}} {\bf{28}} (1998) 225
\bibitem{Riess}
A. G. Riess et al., ``Observational Evidence from Supernovae for an Accelerating Universe and a Cosmological Constant'', {\it{Astron. J.}} {\bf{116}} (1998) 1009
\bibitem{Perlmutter}
S. Perlmutter et al., ``Measurements of Omega and Lambda from $42$ high redshift supernovae'', {\it{Astrophys. J.}} {\bf{517}} (1999) 
565 
\bibitem{Simm}
J. G. Simmonds and J. E. Mann, {\it{A first look at perturbation theory}}, Dover Publications, 1998
\bibitem{Pod}
J. B. Griffiths and J. Podolsky, {\it{Exact Space-Times in Einstein's General Relativity}}, Cambridge University Press, 2009
\bibitem{Brill}
D. R. Brill and S. A. Hayward, ``Global Structure of a Black-Hole Cosmos and its Extremes'', {\it{Class. Quant. Grav.}} {\bf{11}} (1994) 359
\bibitem{me1} I. Arraut, D. Batic and  M. Nowakowski, ``Comparing two approaches to Hawking radiation of Schwarzschild-de Sitter black holes'', 
Class. Quant. Grav. {\bf 26} (2009) 125006
\bibitem{me2} M. Nowakowski, `` The consistent Newtonian limit of Einstein's gravity with a cosmological constant'' Int. J. Mod. Phys. , 
{\bf 10} (2001) 649
\bibitem{me3}  A. Balaguera-Antolinez, C. G. Boehmer and M. Nowakowski, ``Scales set by the cosmological constant'', Class.Quant. Grav. {\bf28} (2006) 485
\bibitem{me4} A. Balaguera-Antolinez, C.G. Boehmer and M. Nowakowski, ``On astrophysical bounds of the cosmological constant'', Int. J. Mod. Phys. {\bf D14} (2005)
1507
\bibitem{me5} M. Nowakowski, J.C. Sanabria and A. Garcia, ``Gravitational equilibrium in the presence of a positive cosmological constant'', Phys. Rev. {\bf D66} 
(2002) 023003
\bibitem{Page} 
D. Page, ``Dirac equation around a charged, rotating black hole'', {\it{Phys. Rev. D}} {bf{14}} (1976) 1509
\bibitem{New} 
E. T. Newman and R. Penrose, ``An Approach to Gravitational Radiation by a Method of Spin Coefficients'', {\it{J. Math. Phys.}} {\bf{3}} (1962) 566
\bibitem{pen}
R. Penrose and W. Rindler, {\it{Spinors and Space-Time}}, Vol. $1$, Cambridge University Press, 1986
\bibitem{Kinn} 
W. Kinnersley, ``Type D Vacuum Metrics'', {\it{J. Math. Phys.}} {\bf{10}} (1969) 1195
\bibitem{Carter}
B. Carter, {\it{Gravitation in Astrophysics}}, NATO ASI Series B, vol. 156, Plenum Press, 1987
\bibitem{Chandra}
S. Chandrasekhar, ``The Solution of Dirac's Equation in Kerr Geometry'', {\it{Proc. R. Soc. London}} {\bf{349}} (1976) 571
\bibitem{Davide2}
D. Batic, H. Schmid, and M. Winklmeier, ``On the Eigenvalues of the Chandrasekhar-Page Angular Equation'', {\it{J. Math. Phys.}} {\bf{46}} (2005) 012504
\bibitem{Davide3}
D. Batic and M. Sandoval, ``The hypergeneralized Heun equation in quantum field theory in curved space-times'', {\it{Cent. Eur. J. Phys.}}  {\bf{8}} (2009) 490
\bibitem{S1}
R. Sch\"afke, ``A connection problem for a regular and an irregular point of complex ordinary differential equations'', {\it{SIAM J. Math. Anal.}} {\bf{15}} (1984) 253
\bibitem{S2}
R. Sch\"afke and D. Schmidt, ``The connection problem for general linear ordinary differential equations at two regular 
singular points with applications in the theory of special functions'', {\it{SIAM J. Math. Anal.}} {\bf{11}} (1980) 848
\bibitem{Byrd}
P. F. Byrd and M. D. Friedman, {\it{Handbook of elliptic integrals for engineers and physicists}}, Springer Verlag, 1954
\bibitem{Weidmann1}
J. Weidmann, {\it{Linear Operators in Hilbert Spaces}}, Springer Verlag, 1980
\bibitem{Weidmann2}
J. Weidmann, {\it{Spectral Theory of ordinary differential operators}}, vol. 1258, Lecture Notes in Mathematics, Springer Verlag, 1987
\bibitem{Kato} 
T. Kato, {\it{Perturbation Theory for Linear Operators}}, Springer Verlag Berlin, Heidelberg, New York, second edition, 1980
\bibitem{Wink}
M. Winklmeier, {\it{The Angular Part of the Dirac equation in the Kerr-Newman metric: Estimates for the eigenvalues}}, Verlag Dr. Hut, 2006
\bibitem{Lesch}
M. Lesch and M. Malamud, ``On the deficiency indices and self-adjointness of symmetric Hamiltonian systems'', {\it{J. Differential Equations}} {\bf{189}} (2005) 556
\bibitem{BN}
D. Batic and M. Nowakowski, ``On the bound states of the Dirac equation in the extreme Kerr metric'', {\it{Class. Quant. Grav.}} {\bf{25}} (2008) 225022
\bibitem{Neumark}
M. A. Neumark, {\it{Lineare Differentialoperatoren}}, Akademie Verlag: Berlin, 1960
\bibitem{Coppel} 
W. A. Coppel, {\it{Stability and Asymptotic Behavior of Differential Equations}}, Heath Publishing Company, 1965
\bibitem{Zhang}
Y. Z. Zhang, ``Exact polynomial solutions of second order differential equations and their applications'', {\it{J. Phys. A: Math. and Theor.}} {\bf{45}} (2012) 065206

\end{thebibliography}
\end{document}